\documentclass[11pt]{article}
\usepackage{jheppub}
\usepackage{bm,bbm}
\usepackage{booktabs,amsmath}
\usepackage{accents}
\usepackage{multirow}
\usepackage{graphics}
\usepackage{braket}
\usepackage{mathtools}
\usepackage{comment}
\usepackage{autobreak}
\usepackage{subcaption}
\usepackage{enumerate}
\usepackage{blkarray}
\usepackage{longtable}

\usepackage{adjustbox}

\usepackage {framed,color}

\usepackage{tikz}
\usetikzlibrary{patterns}
\usetikzlibrary{arrows,shapes,positioning}
\usetikzlibrary{decorations.markings}
\usetikzlibrary{decorations.pathmorphing}
\tikzset{snake it/.style={decorate, decoration=snake}}

\hypersetup{
    colorlinks=true,
    linkcolor=blue,
    filecolor=magenta,      
    urlcolor=cyan,
    pdfpagemode=FullScreen,
    }

\usepackage{arydshln}
\usepackage{mathtools}

\edef\restoreparindent{\parindent=\the\parindent\relax}
\usepackage{parskip}
\restoreparindent
\usepackage{ascmac}
\usepackage{mathrsfs}

\usepackage{amsthm}
\newtheoremstyle{break}
  {\topsep}{\topsep}%
  {\upshape}{}%
  {\bfseries}{}%
  {\newline}{}%
\theoremstyle{break}
\newtheorem{theorem}{Theorem}[section]
\newtheorem{proposition}{Proposition}[section]

\newtheorem*{definition}{Definition}

\newtheorem{conjecture}{Conjecture}[section]

\def\Tr{{\rm Tr}}
\def\d{{\rm d}}
\def\i{{\rm i}}

\def\CH{{\cal H}}

\def\CN{{\cal N}}
\def\CO{{\cal O}}
\def\CP{{\cal P}}

\def\CX{{\cal X}}

\def\BC{\mathbb{C}}
\def\BF{\mathbb{F}}

\def\BR{\mathbb{R}}

\def\BZ{\mathbb{Z}}

\def\EU{\mathscr{U}}

\def\d{\mathrm{d}}
\def\SO{\mathrm{SO}}

\title{
Elliptic genera from classical error-correcting codes}

\author[a,b]{Kohki Kawabata}
\author[a]{and Shinichiro Yahagi}

\affiliation[a]{Department of Physics, Faculty of Science,
The University of Tokyo,\\
Bunkyo-Ku, Tokyo 113-0033, Japan}

\affiliation[b]{Department of Physics, Osaka University,\\
Machikaneyama-Cho 1-1, Toyonaka 560-0043, Japan}

\preprint{OU-HET-1200}

\abstract{
We consider chiral fermionic conformal field theories constructed from classical error-correcting codes and provide a systematic way of computing their elliptic genera.
We exploit the $\mathrm{U}(1)$ current of the $\mathcal{N}=2$ superconformal algebra to obtain the $\mathrm{U}(1)$-graded partition function that is invariant under the modular transformation and the spectral flow.
We demonstrate our method by constructing extremal $\mathcal{N}=2$ elliptic genera from classical codes for relatively small central charges. 
Also, we give near-extremal elliptic genera and decompose them into $\mathcal{N}=2$ superconformal characters.

}

\begin{document}
\maketitle
\flushbottom

\newpage

\section{Introduction}

Recently, the construction of conformal field theories (CFTs) from error-correcting codes has been revisited and extended to fermionic CFTs with supersymmetry.
It has been widely known that error-correcting codes can be exploited to construct CFTs for a long time.
The program initially began with chiral bosonic CFTs constructed from classical error-correcting codes~\cite{frenkel1984natural,frenkel1989vertex,Dolan:1994st}, driven by insights from the understanding of the Monstrous moonshine~\cite{conway1979monstrous}.
Afterward, via fermionization, some of their fermionic counterparts have been shown to have supersymmetry~\cite{Dixon:1988qd,Benjamin:2015ria,Harrison:2016hbq,Moore:2023zmv}.
These pioneering developments have inspired the direct construction of supersymmetric chiral CFTs from classical ternary codes~\cite{Gaiotto:2018ypj} and its generalization to classical $p$-ary codes~\cite{Kawabata:2023nlt}.\footnote{Although this paper focuses on chiral CFTs constructed using classical error-correcting codes, there is another approach that builds nonchiral CFTs from quantum error-correcting codes~\cite{Dymarsky:2020qom,Kawabata:2022jxt,Alam:2023qac}. Recently, the direction has taken significant advancement in various aspects, including the modular bootstrap program~\cite{Dymarsky:2020bps,Henriksson:2022dnu,Dymarsky:2022kwb} and the relation to holographic duality~\cite{Dymarsky:2020pzc,Henriksson:2021qkt,Angelinos:2022umf}. See~\cite{Dymarsky:2021xfc,Buican:2021uyp,Furuta:2022ykh,Yahagi:2022idq,Henriksson:2022dml,Furuta:2023xwl,Kawabata:2023usr,Kawabata:2023iss} for other related progress.}
The chiral fermionic CFTs, termed fermionic code CFTs, have demonstrated their utilities in verifying conjectured constraints on superconformal indexes~\cite{Gaiotto:2018ypj} and searching for theories with large spectral gaps~\cite{Kawabata:2023nlt}.

The construction of fermionic code CFTs exploits the relationship between classical error-correcting codes and Euclidean lattices.
Concretely, a code can specify a lattice by the method called ``Construction A"~\cite{leech1971sphere}.
Then the lattice can be regarded as a momentum lattice in a CFT and provides a set of vertex operators in the Neveu-Schwarz (NS) sector.
The Ramond (R) sector can be obtained from the shadow of the lattice, which is a set of points shifted by half of a certain vector, called the characteristic vector \cite{conway2013sphere}.

This paper aims to compute the elliptic genera of fermionic code CFTs.
In supersymmetric quantum field theories, the Witten index~\cite{Witten:1982df,Witten:1982im} plays an essential role in characterizing a theory through a topological invariant under continuous deformations preserving supersymmetry.
For $\CN=2$ superconformal theories, one can refine the Witten index to the elliptic genus~\cite{Witten:1986bf} by inserting the U(1) charge.
The elliptic genera are invariant under the modular transformation and the spectral flow and are described as weak Jacobi forms in number theory.
For a fermionic code CFT with central charge $n$, the Witten index can be computed by the partition function twisted by the fermion parity $(-1)^F$ in the Ramond sector $\CH_\mathrm{R}$~\cite{Gaiotto:2018ypj,Kawabata:2023nlt}
\begin{align}
    \Tr_{\CH_\mathrm{R}} \left[(-1)^F \,q^{L_0-\frac{n}{24}} \right]\,,
\end{align}
where $q=e^{2\pi\i\tau}$, $\tau$ is the torus modulus, and $L_0$ is the zero mode of the Virasoro generators.
To compute the elliptic genus, we extend the Witten index to the partition function graded by U(1) charge $J_0$
\begin{align}
\label{eq:graded_intro}
    \Tr_{\CH_\mathrm{R}} \left[(-1)^F \,q^{L_0-\frac{n}{24}}\,y^{J_0} \right]\,,
\end{align}
where $y=e^{2\pi\i z}$ and $z\in\BC$.

Then, we obtain the following main theorem (Theorem~\ref{theorem:weak}):
If a classical code $C$ satisfies
\begin{itemize}
    \item The Construction A lattice $\Lambda(C)$ is an odd self-dual lattice of rank $n\in12\BZ$.
    \item The minimum norm of characteristic vectors in $\Lambda(C)$ is $n/3$.
\end{itemize}
the graded torus partition function~\eqref{eq:graded_intro} of a fermionic code CFT becomes a weak Jacobi form by taking an appropriate U(1) charge $J_0$.

To extend torus partition functions to elliptic genera, we identify the U(1) current of the $\CN=2$ superconformal algebra in fermionic code CFTs.
Since elliptic genera are torus partition functions with the insertion of U(1) current describing the R-symmetry, we need to identify the R-symmetry in fermionic code CFTs.
Fermionic code CFTs consist of free chiral bosons and have U(1) currents generated by the free bosons.
We take a linear combination of them and impose some constraints to describe the R-symmetry.
We see that the constraints are satisfied if one determines a linear combination by a lattice vector called a characteristic vector~\cite{conway2013sphere,serre2012course,milnor1973symmetric}.
The corresponding U(1)-graded partition functions show the modular invariance and spectral flow invariance for the central charge being a multiple of $12$ (Proposition~\ref{prop:ell_trans}).
Furthermore, when satisfying an additional condition, the U(1)-graded partition functions turn out to be weak Jacobi forms (Theorem~\ref{theorem:weak}).
With these observations in hand, we investigate the $\CN=2$ structure of fermionic code CFTs.
One characteristic of $\CN=2$ superconformal algebra is the spectral flow connecting a periodicity condition to a different one (e.g., the anti-periodic one to the periodic one)~\cite{Schwimmer:1986mf}.
In the Hilbert space of a fermionic code CFT, we show that the spectral flow can be regarded as a shift of a momentum lattice by a characteristic vector, and an operator acting as the spectral flow arises as a vertex operator whose momentum is a characteristic vector.
Also, we consistently identify the BPS states known as chiral primary states and show that they obey a strong constraint between their conformal weights and U(1) charges~\cite{Lerche:1989uy} in fermionic code CFTs.

We demonstrate our computation of elliptic genera through the construction of extremal $\CN=2$ elliptic genera from fermionic code CFTs.
The notion of extremal CFTs was introduced to explore potential holographic duals to pure AdS$_3$ gravity~\cite{Witten:2007kt} (see also~\cite{hohn2008conformal,Jankiewicz:2006mv}).
It was generalized to theories with $\CN=2$ supersymmetry by defining extremal elliptic genera~\cite{Gaberdiel:2008xb}, with some of them realized using actual CFTs with relatively small central charges~\cite{Cheng:2014owa,Benjamin:2015ria}.
However, to our best knowledge, these extremal elliptic genera have been found independently, and there has not been a unified construction yet.
In this paper, we uniformly reproduce the extremal $\CN=2$ elliptic genera with the central charges $12,24$ in~\cite{Cheng:2014owa,Benjamin:2015ria} by employing our computation with classical codes.
We also construct other elliptic genera from classical codes and interpret them as near-extremal elliptic genera, which obey a weaker condition on extremality~\cite{Gaberdiel:2008xb}.
Moreover, we decompose these elliptic genera into $\CN=2$ superconformal characters and discuss their implications.

The organization of this paper is as follows.
In section~\ref{sec:elliptic_review}, we review $\CN=2$ superconformal field theories and their elliptic genera.
Taking the spectral flow of importance, we explain the Ramond ground states, chiral primary states, and their bounds on conformal weights and U(1) charges.
Also, we introduce elliptic genera as weak Jacobi forms and define their polar regions and extremal elliptic genera, which we will use in later sections.
Section~\ref{sec:elliptic_code} is the central part of the present paper.
Starting with a brief review of fermionic code CFTs, we develop the computation of their elliptic genera.
We obtain partition functions that return the elliptic genera for $\CN=2$ theories.
To compute elliptic genera concretely, we explain how to choose an appropriate U(1) current in fermionic CFTs constructed from $p$-ary codes.
Also, we identify the spectral flow operator and chiral primary states in fermionic code CFTs, and we see that, especially in the ternary case $(p=3)$, those expressions become much simple.
In section~\ref{sec:extremal}, we exemplify the calculation of elliptic genera by reproducing extremal elliptic genera.
We deduce classical $p$-ary codes that yield extremal elliptic genera for small prime $p$ and obtain fermionic code CFTs with extremal elliptic genera for central charges $12,24$.
Section~\ref{sec:near-extremal} deals with other examples of elliptic genera.
Considering some classical codes, we give their elliptic genera that admit an $\CN=2$ character decomposition and discuss their near-extremality.
In section~\ref{sec:discussion}, we summarize the results and their implications.
Appendix~\ref{app:N2_character} aligns representations of the $\CN=2$ algebra and their characters.

\section{Elliptic genera and weak Jacobi forms}
\label{sec:elliptic_review}

In this section, we review $\CN=2$ supersymmetric CFTs with special attention to their elliptic genera following \cite{Lerche:1987ca,Lerche:1989uy,Kawai:1993jk,Dijkgraaf:2000fq,Moore:2004fg,Manschot:2007ha,Gaberdiel:2008xb,Manschot:2008zb}. 
Section~\ref{ss:N=2} is devoted to introducing the $\CN=2$ superconformal algebra and its spectral flow.
We also review $\CN=2$ structures such as the Ramond ground states and chiral primary fields.
Noting that the elliptic genera are invariant under the modular transformation and spectral flow, we describe them with characterization as weak Jacobi forms in section~\ref{ss:weak_Jacobi}.
In section~\ref{ss:extremal_intro}, we introduce extremal and near-extremal elliptic genera defined in~\cite{Gaberdiel:2008xb}, which we will reproduce from classical error-correcting codes in section~\ref{sec:extremal} and section~\ref{sec:near-extremal}, respectively.

\subsection{$\CN=2$ superconformal symmetry}
\label{ss:N=2}

An $\CN=2$ superconformal theory has stress tensor $T(z)$, two supercurrents $G^1(z)$, $G^2(z)$, and $\SO(2)\cong \mathrm{U}(1)$ current $J(z)$ describing the R-symmetry.
It is convenient to define
\begin{align}
    G^\pm(z) = \left(G^1(z) \pm \i \,G^2(z)\right)/\sqrt{2}\,.
\end{align}
Then, the supercurrents $G^\pm(z)$ have charge $Q = \pm 1$ under the U(1) current $J(z)$. 
For an $\CN=2$ theory with the central charge $n$, these operators have the following operator product expansions (OPEs):
\begin{align}
    \begin{aligned}
    \label{eq:N=2SCA}
    T(z)T(w) &\sim \frac{\frac{n}{2}}{(z-w)^4} + \frac{2}{(z-w)^2}\,T(w) + \frac{1}{z-w}\,\partial T(w)\,,\\
    T(z)G^\pm(w) &\sim \frac{\frac{3}{2}}{(z-w)^2}\,G^\pm(w) + \frac{1}{z-w}\,\partial G^\pm(w)\,,\\
    T(z)J(w) &\sim \frac{J(w)}{(z-w)^2} + \frac{\partial J(w)}{z-w}\,,\qquad
    G^\pm(z)G^\pm(w) \sim 0\,,\\
    J(z)J(w) &\sim \frac{\frac{n}{3}}{(z-w)^2}\,,\qquad
    J(z)G^\pm(w) \sim \pm \frac{G^\pm(w)}{z-w}\,,\\
    G^+(z)G^-(w) &\sim \frac{\frac{2}{3}n}{(z-w)^3} +\frac{2J(w)}{(z-w)^2} + \frac{2T(w) + \partial J(w)}{z-w}\,.
    \end{aligned}
\end{align}
The mode expansions of the operators are
\begin{align}
    T(z) = \sum_{m\,\in\,\BZ}\frac{L_m}{z^{m+2}}\,,\qquad
    G^\pm(z) = \sum_r\frac{G^\pm_r}{z^{r+\frac{3}{2}}}\,,\qquad
    J(z) = \sum_{m\,\in\,\BZ} \frac{J_m}{z^{m+1}}\,,
\end{align}
where $r\in \BZ+1/2$ in the Neveu-Schwarz (NS) sector and $r\in\BZ$ in the Ramond (R) sector.
Accordingly, the commutation relations of the $\CN=2$ algebra are
\begin{align}
\begin{aligned}
    \left[L_m,L_k\right] &= (m-k)\,L_{m+k} + \frac{n}{12}(m^3-m)\,\delta_{m+k,0}\,,\\
    \left[L_m,G^\pm_r\right] &= \left(m/2-r\right)G^\pm_{m+r}\,,\qquad
    \left[L_m,J_k\right] = -k J_{m+k}\,,\\
    \left[J_m,J_k\right] &= (n/3)\,m\,\delta_{m+k,0}\,,\qquad\;
    \left[J_m,G^\pm_r\right] = \pm G_{m+r}^\pm\,,\\
    \left\{G^+_r,G^-_s\right\} &= 2\,L_{r+s} + (r-s)J_{r+s}+ \frac{n}{3}\left(r^2-\frac{1}{4}\right)\delta_{r+s,0}\,, 
\end{aligned}
\end{align}
and $\{G_r^+,G_s^+\} = \{G_r^-,G_s^-\} = 0$.

The $\CN=2$ superconformal algebra has the spectral flow isomorphism~\cite{Schwimmer:1986mf}. This connects different sectors (e.g., the NS sector and the R sector) by twisting the boundary condition of the operators by the U(1) current $J(z)$.
Using a twist parameter $\theta$, the spectral flow operator acting on the Hilbert space $\CH_0$ is denoted by
\begin{align}
    \EU_\theta : \CH_0 \to \CH_\theta\,,
\end{align}
where $\CH_\theta$ is the Hilbert space quantized on the twisted boundary condition with the parameter $\theta$.
Under the spectral flow, the Virasoro generators and the modes of  the U(1) current and the supercurrents transform by (\cite{Lerche:1989uy})
\begin{align} 
\begin{aligned}
\label{eq:sf_LJ}
    \EU_\theta\, L_m \,\EU_\theta^{-1} &= L_m + \theta J_m + \frac{n}{6}\theta^2\delta_{m,0}\,,&\quad 
    \EU_\theta \,J_m\, \EU_\theta^{-1} &= J_m + \frac{n}{3} \theta \delta_{m,0}\,,\\
    \EU_\theta \,G_r^+ \,\EU_\theta^{-1} &= G_{r+\theta}^+\,,& 
    \EU_\theta\, G_r^- \,\EU_\theta^{-1} &= G_{r-\theta}^-\,.
\end{aligned}
\end{align}
Accordingly, a state $\ket{h,Q}\in\CH_0$ with conformal weight $h$ and U(1) charge $Q$ transforms into the twisted state
\begin{align}
\label{eq:spct_qn}
    \EU_\theta \ket{h,Q} = \Ket{h-\theta Q +\frac{n}{6}\theta^2\,,\, Q-\frac{n}{3}\theta}\in\CH_\theta\,.
\end{align}
For $\theta \in \BZ+1/2$, the spectral flow connects the NS and R sectors. Then, the NS sector and the R sector are isomorphic in $\CN=2$ superconformal theories \cite{Schwimmer:1986mf}. For $\theta\in\BZ$, the spectral flow takes NS to NS and R to R.

A supersymmetric theory has a unitarity bound on the conformal weight in the R sector. From the $\CN=2$ algebra, we have the relation
\begin{align}
    \left\{G_0^-,G_0^+\right\} = 2 L_0 - n/12\,.
\end{align}
Since the left-hand side is a positive operator, any state in the R sector must satisfy the bound $h \geq n/24$.
One can show that the bound is saturated $(h = n/24)$ if and only if a state satisfies the condition
\begin{align}
\label{eq:ramond_gr}
    G_m^\pm \ket{\tilde{\phi}} = 0\,,\qquad (m\geq0)\,.
\end{align}
To see this, suppose $\ket{\tilde{\phi}}$ satisfies $h = n/24$, which implies that
\begin{align}
    \bra{\tilde{\phi}} \left\{G_0^-,G_0^+\right\}\ket{\tilde{\phi}} = 0 = \left|G_0^+\ket{\tilde{\phi}}\right|^2 + \left|G_0^-\ket{\tilde{\phi}}\right|^2\,.
\end{align}
Hence we have $G_0^+\ket{\tilde{\phi}} = G_0^-\ket{\tilde{\phi}}=0$.
Note that the operators $J_m$ $(m>0)$, which lower the $L_0$ eigenvalue, must annihilate a state $\ket{\tilde{\phi}}$ saturating $h=n/24$.
Otherwise, the unitarity bound does not hold.
Therefore, using the $\CN=2$ commutation relation
\begin{align}
    \left[J_m, G_0^\pm\right] = \pm G_m^\pm\,,
\end{align}
we obtain $G_m^\pm \ket{\tilde{\phi}}=0$ for $m\geq0$.
These states have the lowest energy in the R sector, so they are called the Ramond ground states.

Since the NS and R sectors are isomorphic, there exist states in the NS sector, which are connected to the Ramond ground states by the spectral flow.
Let us consider the spectral flow with $\theta = -1/2$. Then, the supercurrents transform into
\begin{align}
    \EU_{-1/2}\, G_r^+ \,\EU_{-1/2}^{-1} = G_{r-1/2}^+\,,\qquad \EU_{-1/2}\, G_r^-\, \EU_{-1/2}^{-1} = G_{r+1/2}^-\,.
\end{align}
Using these relations, the condition \eqref{eq:ramond_gr} for the Ramond ground states becomes
\begin{align}
    \EU_{-1/2} \,G_m^\pm\, \EU_{-1/2}^{-1} \,\left(\EU_{-1/2}\, \ket{\tilde{\phi}}\right) = G^\pm_{m\mp1/2}\ket{\phi} = 0 \qquad (m\geq0)\,,
\end{align}
where $\ket{\phi} = \EU_{-1/2} \ket{\tilde{\phi}}$ is in the NS sector.
Hence, the resulting state $\ket{\phi}$ satisfies the conditions
\begin{align}
\label{eq:chiral_primary}
    G^+_{-1/2}\ket{\phi} = 0\,,\qquad G^-_{m+1/2}\ket{\phi} = G^+_{m+1/2}\ket{\phi} = 0\qquad (m\geq0)\,.
\end{align}
These states are called chiral primary states \cite{Lerche:1989uy}, which can be obtained from the Ramond ground states by the spectral flow with $\theta =-1/2$.
In the opposite direction, the spectral flow with $\theta = 1/2$ transforms chiral primary states into the Ramond ground states.
Similarly, anti-chiral primary states are defined by only replacing $G^+$ with $G^-$ in \eqref{eq:chiral_primary} and they are mapped to the Ramond ground states by the spectral flow $\theta =-1/2$.

Chiral primary states are subject to a strong constraint on the conformal weight and the U(1) charge.
From the $\CN=2$ algebra, we have the relation
\begin{align}
    \left\{G^-_{1/2}, G^+_{-1/2}\right\}\ket{\phi} = (2L_0 - J_0)\ket{\phi}\,.
\end{align}
As the left-hand side is a positive operator, any state in the NS sector satisfies the bound
\begin{align}
\label{eq:uni_bound_NS}
    h\geq \frac{Q}{2}\,.
\end{align}
Chiral primary states saturate the bound and have the conformal weight $h = Q/2$.
Similarly, an anti-chiral primary state with charge $Q$ has the conformal weight $h = -Q/2$.
Additionally, we can show that any state with $h = Q/2$ is both chiral and primary.
We can also obtain other bounds on chiral primary states from the $\CN =2$ algebra
\begin{align}
    \left\{G^-_{3/2}, G^+_{-3/2}\right\} = 2L_0 - 3J_0 + 2n/3\,.
\end{align}
Since the left-hand side is a positive operator and any chiral primary state satisfies $h = Q/2$, we have the inequality 
\begin{align}
\label{eq:bound_chiralpr}
    h\leq n/6\,.
\end{align}

Via the state-operator isomorphism, the operators corresponding to chiral primary states are called chiral primary fields.
Using the inequalities~\eqref{eq:uni_bound_NS} and \eqref{eq:bound_chiralpr}, it is straightforward to see that the OPEs between chiral primary fields do not have singular terms and are well-defined without regularizations~\cite{Lerche:1989uy}.
Such OPEs are closed in chiral primary fields. Hence, the space of chiral primary fields is called a chiral ring with the product structure by the operator product.
In section~\ref{ss:spectral_flow}, we identify the spectral flow operator, the Ramond grounds states, and the chiral ring in fermionic CFTs constructed from classical error-correcting codes.

\subsection{Elliptic genera as weak Jacobi forms}
\label{ss:weak_Jacobi}

Let us consider a supersymmetric CFT with the central charge $n=6m$ $(m\in\BZ)$ on the torus with the modulus $\tau=\tau_1+\i\tau_2$.
The Witten index is defined by the Ramond sector partition function with the insertion of the fermion parity $(-1)^F$
\begin{align}
    \Tr_{\CH_\mathrm{R}} \left[(-1)^F \,q^{L_0-\frac{n}{24}}\,\bar{q}^{\bar{L}_0-\frac{n}{24}} \right]\,,
\end{align}
where  $q=e^{2\pi\i\tau}$.
For an $\CN=2$ superconformal theory, the Witten index can be refined to an elliptic genus using the U(1) current $J(z)$ of the $\CN=2$ algebra.\footnote{For another refinement called a mod-2 elliptic genus, see e.g.~\cite{Tachikawa:2023nne}.}
Following~\cite{Witten:1993jg}, we define the elliptic genus by
\begin{align}
    Z_{\mathrm{EG}}(\tau,z) = \Tr_{\CH_\mathrm{R}} \left[(-1)^F \,q^{L_0-\frac{n}{24}}\,\bar{q}^{\bar{L}_0-\frac{n}{24}} \,y^{J_0} \right]\,,
\end{align}
where $J_0$ is the zero mode of $J(z)$ and $y=e^{2\pi\i z}$ with $z\in\BC$.

If the left- and right-moving R charges $(Q, \bar{Q})$ satisfy $Q-\bar{Q}\in\BZ$ for any state, we can define the fermion parity $F = F_L + F_R$ in terms of the U(1) current \cite{Lerche:1989uy,Witten:1993jg,Hori:2003ic} 
\begin{align}
\label{eq:fermion_parity_current}
    (-1)^F = \exp{[\i\pi( J_0-\bar{J}_0)]}
\end{align}
In this paper, we restrict our attention to theories with integral U(1) charges.
Then, the modularity properties together with the spectral flow invariance and the unitarity bound in the R sector can be summarized that the elliptic genus $Z_{\mathrm{EG}}(\tau,z)$ is a weak Jacobi form of weight 0 and index $m$~\cite{Kawai:1993jk}.

\begin{definition}[\cite{Eichler:1985}]
A weak Jacobi form $\phi(\tau,z)$ of weight $w$ and index $m\in\BZ$ satisfies the transformation laws
\begin{align}
\label{eq:weak_Jacobi_md}
    \phi\left(\frac{a\tau+b}{c\tau+d},\frac{z}{c\tau+d}\right) &= (c\tau+d)^w\,e^{2\pi\i m \frac{cz^2}{c\tau+d}}\,\phi(\tau,z)\,,\qquad 
    \begin{pmatrix}
        a & b \\
        c & d 
    \end{pmatrix}\in\mathrm{SL}(2,\BZ)\,,\\ 
    \phi(\tau,z+l\tau+l') &= e^{-2\pi\i m(l^2 \tau+ 2lz)}\,\phi(\tau,z)\,,\qquad l,l'\in\BZ\,.
\label{eq:weak_Jacobi_sf}
\end{align}
where $q=e^{2\pi\i\tau},\ y=e^{2\pi\i z}$ and has a Fourier expansion in terms of the integer coefficients $c(n,l)$
\begin{align}
\label{eq:exp_weakjacob}
    \phi(\tau,z) = \sum_{n\,\geq\,0\,,\,l\,\in\,\BZ} c(n,l)\,q^n y^l\,,
\end{align}
with $c(n,l) = (-1)^w\,c(n,-l)$. 
\end{definition}

The modular transformation \eqref{eq:weak_Jacobi_md} is generated by the $T$ and $S$ transformations:
\begin{alignat}{3}
    T: (\tau,z) &\rightarrow (\tau+1,z) \,,\quad & \phi(\tau+1,z) &= \phi(\tau,z) \,,\quad && \begin{pmatrix} 1 & 1 \\ 0 & 1 \end{pmatrix} \,, \label{eq:weakJ_T}\\
    S: (\tau,z) &\rightarrow \left(-\frac{1}{\tau},\frac{z}{\tau}\right) \,,\quad & \phi\left(-\frac{1}{\tau},\frac{z}{\tau}\right) &= \tau^w e^{2\pi\i m\frac{z^2}{\tau}}\phi(\tau,z) \,,\quad && \begin{pmatrix} 0 & -1 \\ 1 & 0 \end{pmatrix} \,, \label{eq:weakJ_S}
\end{alignat}
where the matrices are in $\mathrm{SL}(2,\BZ)$. The shift of $z$ \eqref{eq:weak_Jacobi_sf} is generated by
\begin{alignat}{2}
    (\tau,z) &\rightarrow (\tau,z+1) \,,\quad && \phi(\tau,z+1)=\phi(\tau,z) \,, \label{eq:weakJ_shift_integer} \\
    (\tau,z) &\rightarrow (\tau,z+\tau) \,,\quad && \phi(\tau,z+\tau)=e^{-2\pi\i m(\tau+2z)} \phi(\tau,z) \,. \label{eq:weakJ_shift_tau}
\end{alignat}
These generators are not independent since the shift by $\tau$ \eqref{eq:weakJ_shift_tau} can be derived by combining the $S$ transformation \eqref{eq:weakJ_S} and the shift by integer \eqref{eq:weakJ_shift_integer}. In addition, if $\phi(\tau,z)$ can be written as \eqref{eq:exp_weakjacob}, then the $T$ transformation \eqref{eq:weakJ_T} and \eqref{eq:weakJ_shift_integer} are automatically satisfied. The coefficients condition $c(n,l)=(-1)^w c(n,-l)$ can be derived from $S^2$, i.e., $(\tau,z)\rightarrow (\tau,-z)$. Thus, to prove that $\phi(\tau,z)$ is a weak Jacobi form, it is enough to check the $S$ transformation \eqref{eq:weakJ_S} and the Fourier expansion \eqref{eq:exp_weakjacob}.

A weak Jacobi form $\phi(\tau,z)$ can be decomposed into a set of vector-valued modular forms $h_\mu(\tau)$ and theta functions $\theta_{m,\mu}(\tau,z)$ with $\mu\in\BZ/2m\BZ$:
\begin{align}
    \phi(\tau,z) = \sum_{\mu\,\in\,\BZ/2m\BZ} \,h_\mu(\tau)\, \theta_{m,\mu}(\tau,z)\,.
\end{align}
where
\begin{align}
    h_\mu(\tau) = \sum_{n\,=\,-\mu^2\;\mathrm{mod}\;4m} c_\mu(n)\,q^{n/4m}\,,\qquad 
    \theta_{m,\mu}(\tau,z) = \sum_{\substack{l\,\in\,\BZ\\ l\,=\,\mu\;\mathrm{mod}\; 2m}} q^{l^2/4m}\, y^l\,,
\end{align}
with $c_\mu(n) = (-1)^{2ml}\,c(\frac{n+l^2}{4m},l)$ and $l=\mu$ mod $2m$.
There has been an attempt to construct a Jacobi form from a prescribed set of coefficients \cite{Manschot:2007ha} (one can refer to \cite{niebur1974construction} as math literature).

For our purposes, it is more convenient to use the following expression. A weak Jacobi form of even weight can also be expressed as a polynomial consisting of four generators~(\cite{Eichler:1985,Gaberdiel:2008xb})
\begin{equation}   
    E_4(\tau) \,,\ E_6(\tau) \,,\ \phi_{0,1}(\tau,z) \,,\ \phi_{-2,1}(\tau,z)
\end{equation}
of weights and indices
\begin{equation}
    (w,m) = (4,0) \,,\ (6,0) \,,\ (0,1) \,,\ (-2,1)\,.
\end{equation}
Here, $E_4$ and $E_6$ are the Eisenstein series
\begin{align}
\begin{aligned}
    E_4(\tau) &= 1 + 240 \sum_{n=1}^\infty \frac{n^3 q^n}{1-q^n} = 1 + 240q + 2160q^2 + \CO(q^3) \,, \\
    E_6(\tau) &= 1 - 504 \sum_{n=1}^\infty \frac{n^5 q^n}{1-q^n} = 1 - 504q - 16632q^2 + \CO(q^3)\,,
\end{aligned}
\end{align}
and $\phi_{0,1}$ and $\phi_{-2,1}$ are given by
\begin{align}
\begin{aligned}
    \phi_{0,1}(\tau,z) &= 4 \left( \frac{\theta_2(\tau,z)^2}{\theta_2(\tau,0)^2} + \frac{\theta_3(\tau,z)^2}{\theta_3(\tau,0)^2} + \frac{\theta_4(\tau,z)^2}{\theta_4(\tau,0)^2} \right) \\
    &= y + 10 + y^{-1} + (10y^2 - 64y + 108 - 64y^{-1} + 10y^{-2})q + \CO(q^2) \,, \\[0.2cm]
    \phi_{-2,1}(\tau,z) &= \frac{\theta_1(\tau,z)^2}{\eta(\tau)^6} \\
    &= y - 2 + y^{-1} + (-2y^2 + 8y - 12 + 8y^{-1} - 2y^{-2})q + \CO(q^2) \,.
\end{aligned}
\end{align}
More explicitly, a weak Jacobi form $\phi(\tau,z)$ of weight $w\in2\BZ$ and index $m\in\BZ$ can be written as
\begin{equation} \label{eq:generator_decomposition}
    \phi(\tau,z) = \sum\, c_{q,r,s,t} \,(E_4)^q \,(E_6)^r\, (\phi_{0,1})^s\, (\phi_{-2,1})^t \,,\quad c_{q,r,s,t}\in\BR
\end{equation}
where the sum is taken for $q,r,s,t\in\BZ_{\geq 0}$ such that $4q+6r-2t=w$ and $s+t=m$. When $w=0$, the number of terms is $\mathrm{round}((m+3)^2/12)$ where $\mathrm{round}(x)=\lfloor x+\frac{1}{2} \rfloor$.

According to \cite{Gaberdiel:2008xb}, the polar region of index $m$ is defined by
\begin{equation}
    \CP^{(m)} = \left\{ (n,l) \in \BZ^2 \mid 1 \leq l \leq m \,,\ 0 \leq n \,,\ p:=4mn-l^2 < 0 \right\} \,.
\end{equation}
The term $q^ny^l$ such that $(n,l)\in\CP^{(m)}$ is called the polar term. For example, the polar terms of small even $m$ are
\begin{alignat}{2}
    m=2 &: y^2, y \quad &&\cdots 2 \text{ terms,} \\
    m=4 &: y^4, y^3, y^2, y \quad &&\cdots 4 \text{ terms,} \\
    m=6 &: y^6, y^5, y^4, y^3, y^2, y, qy^6, qy^5 \quad &&\cdots 8 \text{ terms.}
\end{alignat}
A weak Jacobi form $\phi(\tau,z)$ of weight $0$ and index $m$ is uniquely determined by the coefficients of the polar terms since the number of the polar terms of index $m$ is greater than or equal to the number of the independent terms in \eqref{eq:generator_decomposition} \cite{Gaberdiel:2008xb,Keller:2020rwi}.

\subsection{Extremal elliptic genera}
\label{ss:extremal_intro}
The NS vacuum irreducible character of the $\CN=2$ algebra at the central charge $n=6m$ is
\begin{align}
\begin{aligned}
    \mathrm{ch}_{\mathrm{vac}}^{(m)}(\tau,z) &= \Tr_{V_{0}} \left[ q^{L_0-\frac{n}{24}}\, e^{2\pi\i z J_0}\right] \\
    &= q^{-\frac{m}{4}} (1-q) \prod_{k=1}^\infty \frac{(1+yq^{k+\frac{1}{2}})(1+y^{-1}q^{k+\frac{1}{2}})}{(1-q^k)^2} \,,
\end{aligned}
\end{align}
where $V_0$ denotes the vacuum Hilbert space. Using spectral flow, we get the corresponding terms for the R sector as
\begin{align} \label{eq:R_vacuum}
\begin{aligned}
    &\sum_{\theta\in\frac{1}{2}+\BZ} \mathrm{SF}_{\theta} \, \mathrm{ch}_{\mathrm{vac}}^{(m)}(\tau,z+\tfrac{1}{2}) \\
    &:= \sum_{\theta\in\frac{1}{2}+\BZ} \Tr_{V_{0}} \left[\EU_{\theta} \, q^{L_0-\frac{n}{24}} \,e^{2\pi\i (z+\frac{1}{2})J_0} \, \EU_{\theta}^{-1} \right]\\
    &\;= (-1)^m (1-q)\, y^m \prod_{k=1}^\infty \frac{(1-yq^{k+1})(1-y^{-1}q^k)}{(1-q^k)^2} + (\text{Non-polar terms}) \,.
\end{aligned}
\end{align}
Note that only $\mathrm{SF}_{1/2}\, \mathrm{ch}_{\mathrm{vac}}^{(m)}$ is explicitly written since all terms from $\theta\neq 1/2$ are non-polar.
According to \cite{Gaberdiel:2008xb}, we define the extremal elliptic genus as the ``closest" weak Jacobi form to $\mathrm{SF}_{1/2}\, \mathrm{ch}_{\mathrm{vac}}^{(m)}$. More precisely, we call a weak Jacobi form $\phi(\tau,z)$ of weight $0$ and index $m$ as the extremal elliptic genus when the coefficients of the polar terms are equal to $\mathrm{SF}_{1/2}\, \mathrm{ch}_{\mathrm{vac}}^{(m)}$.

The polar terms of $\mathrm{SF}_{1/2}\, \mathrm{ch}_{\mathrm{vac}}^{(m)}$ at small even $m$ are
\begin{align}
    m=2 &: y^2 + 0 y \,, \label{eq:polar_m2} \\ 
    m=4 &: y^4 + 0 (y^3 + y^2 + y) \,, \label{eq:polar_m4} \\
    m=6 &: y^6 + 0 (y^5 + y^4 + y^3 + y^2 + y) + q (y^6 - y^5) \label{eq:polar_m6}
\end{align}
where the polar terms with coefficients $0$ are also shown. As already mentioned, we can uniquely determine the weak Jacobi form from these polar terms. For example, at $m=2$, a weak Jacobi form of weight 0 can be decomposed into $\phi_{0,1}^2$ and $\phi_{-2,1}^2 E_4$ as \eqref{eq:generator_decomposition}, and \eqref{eq:polar_m2} is satisfied when
\begin{equation} \label{eq:extremal_m2}
\begin{aligned}
    Z_{\mathrm{ext}}^{(m=2)}(\tau,z) &= \frac{1}{6} \phi_{0,1}^2 + \frac{5}{6} \phi_{-2,1}^2 E_4 \\
    &= y^2 + 22 + y^{-2} + \CO(q) \,.
\end{aligned}
\end{equation}
Similarly, at $m=4$, \eqref{eq:polar_m4} is satisfied when
\begin{equation} \label{eq:extremal_m4}
\begin{aligned}
    Z_{\mathrm{ext}}^{(m=4)}(\tau,z) &= \frac{1}{432} \phi_{0,1}^4 + \frac{1}{8} \phi_{0,1}^2 \phi_{-2,1}^2 E_4 + \frac{11}{27} \phi_{0,1} \phi_{-2,1}^3 E_6 + \frac{67}{144} \phi_{-2,1}^4 E_4^2 \\
    &=  y^4 + 46 + y^{-4} + \CO(q) \,.
\end{aligned}
\end{equation}
Note that such a weak Jacobi form may not exist for general $m$, and even if it does, there may be no actual CFT that realizes the extremal elliptic genus. In fact, at $m=6$, we cannot get \eqref{eq:polar_m6} from \eqref{eq:generator_decomposition} with any coefficients, which means the non-existence of the extremal elliptic genus.

Therefore, a relaxed definition of near-extremal elliptic genus is also useful. The $\beta$-truncated polar region for $\beta>0$ is defined by
\begin{equation}
\label{eq:near_extremal}
    \CP^{(m)} = \left\{ (n,l) \in \BZ^2 \mid 1 \leq l \leq m \,,\ 0 \leq n \,,\ 4mn-l^2 \leq -\beta \right\} \,.
\end{equation}
Following~\cite{Gaberdiel:2008xb}, we call a weak Jacobi form $\phi(\tau,z)$ of weight $0$ and index $m$ as a $\beta$-extremal elliptic genus when the coefficients of the terms in the $\beta$-truncated polar region are equal to $\mathrm{SF}_{1/2}\, \mathrm{ch}_{\mathrm{vac}}^{(m)}$. Note that a $\beta$-extremal elliptic genus is not unique for general $\beta$.

\section{Elliptic genera from classical codes}
\label{sec:elliptic_code}

In this section, we provide a way of computing the elliptic genera of fermionic code CFTs.
Section~\ref{ss:fermionic_code} briefly reviews the construction of fermionic CFTs from classical $p$-ary codes.
Choosing the U(1) current consistent with the $\CN=2$ algebra from fermionic code CFTs in section~\ref{ss:constraint}, we obtain the U(1)-graded partition function that is invariant under the modular transformation and the spectral flow in section~\ref{ss:ellipticgenera_code}.
Also, we identify the spectral flow operator and chiral ring of fermionic code CFTs in section~\ref{ss:spectral_flow}.

\subsection{Fermionic code CFTs}
\label{ss:fermionic_code}

In this subsection, we review the construction of fermionic CFTs from classical codes over finite fields $\BF_p=\{0,1,\cdots,p-1\}$ of prime order $p$ \cite{Kawabata:2023nlt} (see~\cite{Gaiotto:2018ypj} for the ternary case $(p=3$)).
A $p$-ary linear code of length $n$ can be defined as a subspace $C$ of $\BF_p^n$ and the dimension of $C$ is denoted by $k\leq n$ (see e.g.~\cite{macwilliams1977theory,thompson1983error,welsh1988codes,conway2013sphere,elkies2000lattices1,elkies2000lattices,justesen2004course,nebe2006self} for more details).
It is convenient to introduce a generator matrix $G$ to define a linear code.
A generator matrix $G$ is a $k\times n$ matrix of rank $k$ and generates a linear code
\begin{align}
    C = \left\{ \,c \in \BF_p^n\;\middle|\; c = x \,G \,,\;\;x\in\BF_p^k\,\right\}\,.
\end{align}
The inner product between elements $v=(v_1,\cdots,v_n)$, $w=(w_1,\cdots,w_n)\in\BF_p^n$ is given by the standard Euclidean product $v\cdot w = \sum_{i=1}^n v_i w_i\in\BR$. 
Associated with the inner product, we define the dual code of $C$ by
\begin{align}
    C^\perp = \left\{\,v\in\BF_p^n\;\middle|\; c\cdot v = 0 \mod p\,,\;\;c\in C\,\right\}\,.
\end{align}
We call $C$ self-orthogonal if $C\subset C^\perp$. Then, the dimension of $C$ must be $k\leq n/2$. Also, we call $C$ self-dual if $C=C^\perp$. In this case, $C$ must satisfy $k=n/2$, then $n\in2\BZ$. 
For the binary case $(p=2)$, we can further classify self-orthogonal codes into two classes.
A binary linear code $C$ is doubly-even if any $c\in C$ satisfies $c\cdot c\in4\BZ$. On the other hand, $C$ is called singly-even if there exists an element $c\in C$ such that $c\cdot c\in4\BZ+2$.

A linear code $C\subset\BF_p^n$ can construct a Euclidean lattice via Construction A \cite{leech1971sphere,conway2013sphere,nebe2006self}
\begin{align}
    \label{eq:def_constA_lattice}
    \Lambda(C) = \left\{\,\frac{c+p\,m}{\sqrt{p}}\in\BR^n\;\middle|\;c\in C,\;\,m\in\BZ^n\,\right\}\,,
\end{align}
which has the standard Euclidean inner product denoted by $\cdot$.
As in classical codes, the dual lattice of $\Lambda\subset \BR^n$ is 
\begin{align}
    \label{eq:dual_lattice}
    \Lambda^* = \left\{\lambda'\in\BR^n\mid \lambda\cdot \lambda'\in\BZ\, ,\,\;\lambda\in\Lambda\right\}\,.
\end{align}
A lattice is called integral if $\Lambda\subset\Lambda^*$ and self-dual if $\Lambda = \Lambda^*$.
Additionally, for an integral lattice $\Lambda$, we call $\Lambda$ even if any element $\lambda\in\Lambda$ has an even norm $\lambda\cdot\lambda\in2\BZ$, and odd if there exists an element $\lambda\in\Lambda$ such that $\lambda\cdot\lambda\in2\BZ+1$.

The Construction A lattice $\Lambda(C)$ inherits the property of a linear code $C$.
The following propositions are essential to constructing fermionic code CFTs.

\begin{proposition}[\cite{Kawabata:2023nlt}]
For an odd prime $p$, the Construction A lattice $\Lambda(C)$ is odd self-dual if and only if a linear code $C\subset\BF_p^n$ is self-dual.
\end{proposition}

\begin{proposition}[\cite{Kawabata:2023nlt}]
For $p=2$, the Construction A lattice $\Lambda(C)$ is odd self-dual if and only if a linear code $C\subset\BF_2^n$ is singly-even self-dual.
\end{proposition}

These propositions guarantee that odd self-dual lattices can be constructed from self-dual codes for an odd prime $p$ and singly-even self-dual codes for the binary case ($p=2$).
In what follows, we focus on such classes of linear codes and the corresponding odd self-dual lattices to yield fermionic CFTs.
Specifically, we define the NS sector of a fermionic code CFT by using the Construction A lattice $\Lambda(C)$, and the R sector by the shadow of $\Lambda(C)$.

To introduce the shadow of a lattice, consider an element $\chi\in\Lambda(C)\subset\BR^n$ such that 
\begin{align}
\label{eq:chrst_def}
    \chi\cdot \lambda = \lambda\cdot \lambda\mod 2\,,
\end{align}
for any $\lambda\in\Lambda(C)$.
Such an element $\chi\in\Lambda(C)$ is called characteristic \cite{conway2013sphere,serre2012course,milnor1973symmetric}.
A characteristic vector always exists for an odd self-dual lattice and the choice of a characteristic vector is not unique to a lattice.
It is known that all the characteristic vectors have norm $\chi\cdot\chi = n$ mod $8$ \cite{serre2012course,elkies1999lattices}.
Since self-dual codes $C\subset\BF_p^n$ exist only for $n\in2\BZ$, characteristic vectors of the Construction A lattice $\Lambda(C)$ have even norms:
\begin{align}
\label{eq:norm_characteristic}
    \chi\cdot\chi\in2\BZ\,,\qquad \chi\in\Lambda(C)\,.
\end{align}

Using a characteristic vector $\chi\in\Lambda(C)$, we can define the shadow of $\Lambda(C)$ by
\begin{align}
    S(\Lambda(C)) = \Lambda(C)+\frac{\chi}{2} = \left\{ \lambda+\frac{\chi}{2} \;\middle|\; \lambda\in\Lambda(C) \right\}.
\end{align}
The shadow of a lattice is not a lattice since it is not closed under addition.
While there are several choices for a characteristic vector, the shadow itself does not depend on the choice of a characteristic vector $\chi\in\Lambda(C)$.

We construct the NS sector of a fermionic CFT from the Construction A lattice using the lattice construction of chiral CFTs (see e.g.~\cite{Lerche:1988np,kac1998vertex}).\footnote{Recently, the classification program of chiral fermionic CFTs has developed. Up to now, the classification has been completed for central charges less than or equal to 24. See~\cite{BoyleSmith:2023xkd,Rayhaun:2023pgc,Hohn:2023auw}.}
Let $X(z)$ be an $n$-dimensional chiral boson.
We fix the normalization by the OPE $X(z) X(w) \sim -\log (z-w)$. 
To define the NS sector, we associate an element $\lambda\in\Lambda(C)$ of the Construction A lattice to the vertex operator
\begin{align}
    V_{\lambda}(z) = \,:e^{\i\lambda\cdot X(z)}:\,, \qquad \lambda\in\Lambda(C)\,,
\end{align}
where we omit the cocycle factor because it does not matter for our purpose.
The fundamental operators in a fermionic code CFT can be given by
\begin{align}
    L^i(z) = \i\,\partial X^i(z)\,,\qquad T(z) = -\frac{1}{2}\,:\partial X(z)\cdot \partial X(z):\,,\qquad V_{\lambda}(z) = \,:e^{\i\lambda\cdot X(z)}:\,.
\end{align}
Via the state-operator isomorphism, vertex operators $V_{\lambda}(z)$ are mapped to momentum states $\ket{\lambda}$, which are eigenstates of the Virasoro generator $L_0$ with eigenvalue $h = \lambda^2/2$.
Combining with the excitation by the bosonic oscillators $\alpha_k^i$, the NS sector of a fermionic code CFT is
\begin{align}
    \CH_{\mathrm{NS}}(C) = \left\{\left.\alpha_{-k_1}^{i_1}\cdots\alpha_{-k_r}^{i_r}\ket{\lambda}\;\right| \; \lambda\in\Lambda(C)\,,\;r\in\BZ_{\geq0}\right\}\,,
\end{align}
where the oscillators are in negative modes ($k_1,\cdots,k_r> 0$).
The conformal weight of $\alpha_{-k_1}^{i_1}\cdots\alpha_{-k_r}^{i_r}\ket{\lambda}$ is $h = \lambda^2/2 + \sum_{j=1}^r k_j$, which is an integer if $\lambda^2\in 2\BZ$ and a half-integer if $\lambda^2\in 2\BZ+1$.
From the spin-statics theorem, such a state is bosonic when $\lambda^2\in 2\BZ$ and fermionic when $\lambda^2\in 2\BZ+1$. Thus, we can define the fermion parity in the NS sector by
\begin{align}
\label{eq:fermparity_NS}
    (-1)^F = (-1)^{\chi\cdot \lambda}\,,
\end{align}
where $\chi\in\Lambda(C)$ is a characteristic vector satisfying \eqref{eq:chrst_def}.
This motivates us to divide the Construction A lattice into two parts: $\Lambda(C) = \Lambda_0\cup\Lambda_2$ where
\begin{align}
\begin{aligned}
\label{eq:lattice_NS}
    \Lambda_0 &=  \left\{ \lambda\in\Lambda(C) \mid \chi\cdot\lambda = 0 \mod 2 \right\}\,,
    \\
    \Lambda_2 &= \left\{ \lambda\in\Lambda(C) \mid \chi\cdot\lambda = 1 \mod 2 \right\}\,.
\end{aligned}
\end{align}

On the other hand, the R sector can be defined by associating each element in the shadow $S(\Lambda(C))$ to the vertex operator
\begin{align}
    V_{\lambda+\frac{\chi}{2}}(z) = \,:e^{\i(\lambda+\frac{\chi}{2})\cdot X(z)}:\,,\qquad \lambda+\frac{\chi}{2}\in S(\Lambda(C))\,.
\end{align}
The R sector of a fermionic code CFT is
\begin{align}
    \CH_{\mathrm{R}}(C) = \left\{\left.\alpha_{-k_1}^{i_1}\cdots\alpha_{-k_r}^{i_r}\Ket{\lambda+\frac{\chi}{2}}\;\right| \; \lambda+\frac{\chi}{2}\in S(\Lambda(C))\,,\;r\in\BZ_{\geq0}\right\}\,,
\end{align}
where the bosonic oscillators are in negative modes ($k_1,\cdots,k_r> 0$).
In the R sector, the fermion parity is well-defined only up to an overall factor.
In analogy with the NS sector \eqref{eq:fermparity_NS}, we give the fermion parity of $\alpha_{-k_1}^{i_1}\cdots\alpha_{-k_r}^{i_r}\Ket{\lambda+\frac{\chi}{2}}\in \CH_{\mathrm{R}}(C)$ by
\begin{align}
\label{eq:fermparity_R}
    (-1)^F = (-1)^{\chi\cdot(\lambda+\frac{\chi}{2})}\,.
\end{align}
Note that $\chi\cdot(\lambda+\frac{\chi}{2})$ is an integer from \eqref{eq:norm_characteristic}.
Thus, we can naturally divide the shadow $S(\Lambda(C))$ into two parts: $S(\Lambda(C)) = \Lambda_1\cup\Lambda_3$ where
\begin{align}
\begin{aligned}
\label{eq:lattice_R}
    \Lambda_1 &=  \left\{ \,\lambda+\frac{\chi}{2}\in S(\Lambda(C)) \;\middle|\; \chi\cdot\left(\lambda+\frac{\chi}{2}\right) = 0 \mod 2\, \right\}\,,
    \\
    \Lambda_3 &= \left\{ \,\lambda+\frac{\chi}{2}\in S(\Lambda(C)) \;\middle|\; \chi\cdot\left(\lambda+\frac{\chi}{2}\right) = 1 \mod 2\, \right\}\,.
\end{aligned}
\end{align}

Let us place a fermionic code CFT on the torus with the modulus $\tau = \tau_1 + \i\tau_2$. Then, the torus has the spacial and timelike cycles
\begin{align}
    \text{spacial}: w\sim w+2\pi\,,\qquad
    \text{timelike}: w\sim w+2\pi\tau\,,
\end{align}
where $w$ is the cylindrical coordinate.
Depending on the periodicity (correspondingly, the choice of spin structures), we have the four partition functions
\begin{align}
\begin{aligned}
\label{eq:def_partition_function_ff}
    Z_{\mathrm{NS}}(\tau;\Lambda(C)) &= \Tr_{\CH_{\mathrm{NS}}}\left[{q^{L_0-\frac{n}{24}}}\right] =  \frac{1}{\eta(\tau)^n} \sum_{\lambda\,\in\,\Lambda(C)}  q^{\frac{1}{2}\lambda^2}\,, \\
    Z_{\widetilde{\mathrm{NS}}}(\tau;\Lambda(C)) &= \Tr_{\CH_{\mathrm{NS}}}\left[{(-1)^F\,q^{L_0-\frac{n}{24}}}\right] = \frac{1}{\eta(\tau)^n} \sum_{\lambda\,\in\,\Lambda(C)} (-1)^{\chi\cdot\lambda} q^{\frac{1}{2}\lambda^2}\,, \\
    Z_{\mathrm{R}}(\tau;\Lambda(C)) &= \Tr_{\CH_{\mathrm{R}}}\left[{q^{L_0-\frac{n}{24}}}\right] = \frac{1}{\eta(\tau)^n} \sum_{\lambda\,\in\,\Lambda(C)} q^{\frac{1}{2}(\lambda+\frac{\chi}{2})^2}\,, \\
    Z_{\widetilde{\mathrm{R}}}(\tau;\Lambda(C)) &= \Tr_{\CH_{\mathrm{R}}}\left[{(-1)^F\,q^{L_0-\frac{n}{24}}}\right] = \frac{1}{\eta(\tau)^n} \sum_{\lambda\,\in\,\Lambda(C)} (-1)^{\chi\cdot(\lambda+\frac{\chi}{2})} q^{\frac{1}{2}(\lambda+\frac{\chi}{2})^2}\,.
    \end{aligned}
\end{align}
Here, $\eta(\tau)$ is the Dedekind eta function $\eta(\tau) = q^{\frac{1}{24}}\prod_{m=1}^\infty \,(1-q^m)$.
These partition functions can be collectively denoted by
\begin{align}
    \label{eq:collective_part}
    Z^{\alpha,\beta}(\tau;\Lambda(C)) = \frac{1}{\eta(\tau)^n} \sum_{\lambda\,\in\,\Lambda(C)} (-1)^{\beta\chi\cdot(\lambda+\alpha\frac{\chi}{2})} q^{\frac{1}{2}(\lambda+\alpha\frac{\chi}{2})^2}
\end{align}
where $(\alpha,\beta)$ corresponds to each sector by
\begin{align}
\label{eq:sector_notation}
    \mathrm{NS}: (0, 0)\,,\quad \widetilde{\mathrm{NS}}: (0, 1)\,,\quad
    \mathrm{R}: (1, 0)\,,\quad \widetilde{\mathrm{R}}: (1, 1)\,.
\end{align}
Note that $Z^{\alpha,\beta}(\tau;\Lambda(C))=Z^{\alpha',\beta'}(\tau;\Lambda(C))$ for integers $\alpha\equiv\alpha'$, $\beta\equiv\beta' \mod 2$.

On the torus, any modular transformation is generated by the modular $S$ transformation: $\tau\to-1/\tau$ and the modular $T$ transformation: $\tau\to\tau+1$.
Under these modular transformations, the partition functions behave as
\begin{align}
\begin{aligned}
    Z^{\alpha,\beta}(\tau+1;\Lambda(C)) &= e^{\i\pi(3\alpha^2-1)\frac{n}{12}}\, Z^{\alpha,\alpha+\beta+1}(\tau;\Lambda(C))\,,\\
    Z^{\alpha,\beta}(-1/\tau;\Lambda(C)) &= (-1)^{\alpha\beta\frac{n}{2}}\, Z^{\beta,\alpha}(\tau;\Lambda(C))\,.
\end{aligned}
\end{align}
The phases caused by the modular transformations can be canceled by coupling with $2n$ right-moving Majorana-Weyl fermions.
Then, the partition functions of a fermionic code CFT show the expected modular transformation laws.

\subsection{Constraints on U(1) current}
\label{ss:constraint}

We have introduced chiral fermionic CFTs constructed from classical codes using Euclidean lattices.
It has been known that all the fermionic code CFTs have $\CN=2$ supersymmetry in the ternary case $(p=3)$ \cite{Gaiotto:2018ypj}. 
Also, there is strong evidence for the presence of supersymmetry in a class of fermionic code CFTs in other cases $(p\neq3)$ \cite{Kawabata:2023nlt}.
For an $\CN=2$ supersymmetric theory, the elliptic genus can be computed by the partition function with the U(1) charge inserted.
Here, the corresponding U(1) current must be consistent with the $\CN=2$ algebra and describes the R-symmetry.
We need to identify the U(1) current in fermionic code CFTs to compute their elliptic genera.
This section aims to find the U(1) current by solving constraints for a U(1) current to realize the R-symmetry.

By construction, a fermionic code CFT with the central charge $n$ always has the following $\mathrm{U}(1)$ currents:
\begin{align}
    L^i(z) = \i\, \partial X^i(z)\,,\qquad (i=1,2,\cdots,n)\,.
\end{align}
Let us consider a linear combination of these U(1) currents and denote it by
\begin{align}
    J(z) = \sum_{i=1}^n a_i L^i(z) \,,
\end{align}
where $a = (a_1,a_2,\cdots,a_n)\in \BR^n$.
Although $J(z)$ gives a U(1) current in a fermionic code CFT, such a choice of the coefficients does not always realize the R-symmetry.
We figure out some conditions for the $\mathrm{U}(1)$ current to describe the R-symmetry and explicitly construct the U(1) current $J(z)$ satisfying those constraints.

The first constraint comes from the consistency with the $\CN=2$ superconformal algebra~\eqref{eq:N=2SCA}.
Since we have the OPE between a pair of the current $J(z)$
\begin{align}
    J(z) J(w) \sim \frac{a\cdot a}{(z-w)^2}\,,
\end{align}
we require $J(z)$ to satisfy the norm condition
\begin{align}
\label{eq:norm_cond}
    a\cdot a = \frac{n}{3}\,,
\end{align}
where $n$ is the central charge.

Another condition comes from the relationship \eqref{eq:fermion_parity_current} between the U(1) charge $Q$ and the fermion parity $(-1)^F$.
As shown in the previous section, both the NS and R sectors of a fermionic code CFT are built out of the vertex operators
\begin{align}
    \begin{aligned}
        V_{\kappa}(z) = \,:e^{\i\kappa \cdot X(z)}:\,,\qquad \kappa \in
        \begin{dcases}
            \Lambda(C) & (\text{NS sector})\\
            S(\Lambda(C)) & (\text{R sector})\,.
        \end{dcases}
    \end{aligned}
\end{align}
Taking the OPE between the U(1) current $J(z)$ and a vertex operator $V_{\kappa}(w)$, we obtain
\begin{align}
\label{eq:Rcharge_kap}
    J(z) \, V_{\kappa}(w) \sim \frac{a\cdot \kappa}{z-w}\,.
\end{align}
Thus, each vertex operator $V_{\kappa}(z)$ has the charge $Q = a\cdot \kappa$ with respect to the U(1) current $J(z)$.
On the other hand, combining \eqref{eq:fermparity_NS} and \eqref{eq:fermparity_R}, the fermion number of $V_\kappa(z)$ is 
\begin{align}
\label{eq:ferm_parity_kappa}
    F = \chi \cdot \kappa \mod 2\,,
\end{align}
where $\chi\in\Lambda(C)$ is a characteristic vector of $\Lambda(C)$.

Let us assume that the U(1) charge $Q$ of each state is always integral (we check the consistency soon in this subsection).
From \eqref{eq:fermion_parity_current}, the fermion parity $(-1)^F$ is the $\BZ_2$ subgroup of the R-symmetry in $\CN=2$ theories.
Then, the fermion number $F$ and the U(1) charge $Q$ must be related by
\begin{align}
    e^{\pi\i\, Q} = (-1)^F\,,
\end{align}
which results in $Q = F$ mod $2$.
From \eqref{eq:Rcharge_kap} and \eqref{eq:ferm_parity_kappa}, the U(1) charge of $V_\kappa(z)$ is given by $Q = a\cdot \kappa$ and the fermion number is $F = \chi\cdot\kappa$ mod $2$. Therefore, we obtain an additional constraint on the coefficients $a$
\begin{align}
    a\cdot \kappa = \chi\cdot \kappa \mod 2\,,
\end{align}
where $\kappa\in \Lambda(C)$ for the NS sector and $\kappa\in S(\Lambda(C))$ for the R sector.

In the NS sector, the constraint requires the coefficients $a\in\BR^n$ to satisfy
\begin{align}
\label{eq:NS_constraint}
    a\cdot \lambda = \chi\cdot\lambda = \lambda\cdot \lambda \mod 2 \qquad \lambda\in\Lambda(C)\,,
\end{align}
where we used the definition of a characteristic vector $\chi\in\Lambda(C)$.
This requires $a\in\BR^n$ to be a characteristic vector of $\Lambda(C)$. 
Since characteristic vectors constitute a coset of $2\Lambda(C)$ in $\Lambda(C)$: $\chi+2\Lambda(C) = 2 S(\Lambda(C))$ \cite{elkies1999lattices}, the vector $a$ can be written as
\begin{align}
\label{eq:chara_a_gen}
    a = \chi + 2\lambda'\,,\qquad \lambda'\in\Lambda(C)\,.
\end{align}

On the other hand, a vertex operator in the R sector takes the form $V_{\lambda+\frac{\chi}{2}}(z)$ where $\lambda\in\Lambda(C)$.
In this case, the constraint becomes
\begin{align}
    a\cdot \left(\lambda+\frac{\chi}{2}\right) = \chi\cdot \left(\lambda + \frac{\chi}{2}\right)\mod 2 \,,\qquad \lambda\in\Lambda(C)\,.
\end{align}
Assuming the constraint \eqref{eq:chara_a_gen} in the NS sector, the constraint coming from the R sector returns $\chi\cdot\lambda'=0$ mod $2$, from which we conclude
\begin{align}
\label{eq:R_constraint}
    a = \chi + 2\lambda'\,,\qquad \lambda'\in \Lambda_0\,.
\end{align}
Note that $\chi$ is a characteristic vector that fixes the fermion parity.

From the above discussion, we essentially have two constraints \eqref{eq:norm_cond} and \eqref{eq:R_constraint} on the choice of the $\mathrm{U}(1)$ current.
Now we propose to define the U(1) current $J(z)$ by taking $\lambda'=0\in\Lambda_0$.
Accordingly, we fix the coefficients $a\in\BR^n$ by taking the same vector as $\chi$:
\begin{align}
\label{eq:solution}
    a = \chi\,.
\end{align}
Then, to meet the constraints, we only need to choose a characteristic vector $\chi\in\Lambda(C)$ satisfying $\chi\cdot \chi =n/3$.
The corresponding U(1) current is
\begin{align}
\label{eq:cand_R_current}
    J(z) = \sum_{i=1}^n \,\chi_i \,L^i(z)\,.
\end{align}
The Construction A lattice $\Lambda(C)\subset\BR^n$ does not necessarily have a characteristic vector $\chi\in \Lambda(C)$ with $\chi\cdot\chi = n/3$. 
Even if present, such a characteristic vector is not unique to $\Lambda(C)$ in general.
For later convenience, we define a set of characteristic vectors with norm $n/3$ by
\begin{align}
    \CX_{n/3} = \left\{\,\chi\in2S(\Lambda(C))\;\middle|\;\chi\cdot\chi = n/3\,\right\}\,.
\end{align}
Note that this set becomes empty unless $n\in12\BZ$ since $\chi\cdot \chi = n\mod 8$.
Hence, we only have to specify an element $\chi\in\CX_{n/3}$ to define the U(1) current $J(z)$.
We will see that, for the choice of the U(1) current \eqref{eq:cand_R_current}, the U(1)-graded partition function exhibits the expected modular invariance together with the spectral flow invariance.

Once we pick up a characteristic vector $\chi\in\CX_{n/3}$ for defining the fermion parity, the conditions on the $\mathrm{U}(1)$ current are satisfied by choosing the coefficients $a\in\BR^n$ according to \eqref{eq:solution}.
For the choice of $J(z)$, a vertex operator $V_\lambda(z)$ in the NS sector has the U(1) charge $Q = \chi\cdot\lambda$ and $V_{\lambda+\frac{\chi}{2}}(z)$ in the R sector has the U(1) charge $Q = \chi\cdot\left(\lambda+\frac{\chi}{2}\right)$.
These U(1) charges are integral because all the characteristic vectors of the Construction A lattice have even norms $\chi\cdot\chi\in2\BZ$ as in~\eqref{eq:norm_characteristic}.
This is consistent with the assumption of \eqref{eq:fermion_parity_current} that we have used to determine the U(1) current $J(z)$.

Using the fact that the bosonic oscillators are neutral with respect to the U(1) current $J(z)$, the state-operator isomorphism tells us that an NS state $\alpha_{-k_1}^{i_1}\cdots\alpha_{-k_r}^{i_r}\ket{\lambda}\in \CH_{\mathrm{NS}}(C)$ has the following conformal weight and U(1) charge:
\begin{align}
    h = \frac{\lambda^2}{2} + \sum_{j=1}^r k_j \,,\qquad Q = \chi\cdot\lambda\,,
\end{align}
On the other hand, a Ramond state $\alpha_{-k_1}^{i_1}\cdots\alpha_{-k_r}^{i_r}\ket{\lambda+\frac{\chi}{2}}\in \CH_\mathrm{R}(C)$ has 
\begin{align}
\label{eq:qn_r}
    h = \frac{1}{2}\left(\lambda+\frac{\chi}{2}\right)^2 + \sum_{j=1}^r k_j \,,\qquad Q=\chi\cdot \left(\lambda+\frac{\chi}{2}\right)\,.
\end{align}
Here, $k_j$ $(j=1,2,\cdots,r)$ are positive integers.
The next section reveals that the Ramond partition function with the corresponding U(1) charge inserted shows the same modular transformation law as weak Jacobi forms.

\subsection{Elliptic genera of fermionic code CFTs}
\label{ss:ellipticgenera_code}

In this section, based on the choice of U(1) current \eqref{eq:cand_R_current}, we compute the Ramond-Ramond partition function with the corresponding U(1) charge inserted, which can be expected to return elliptic genera for $\CN=2$ superconformal theories.
Furthermore, we show that, under the modular transformation and the spectral flow, the partition functions behave in the same manner as weak Jacobi forms.

Consider a fermionic CFT constructed from a self-dual code $C\subset\BF_p^n$ (singly-even self-dual code for $p=2$) and pick a characteristic vector $\chi\in\CX_{n/3}$.
Using the U(1) current given by \eqref{eq:cand_R_current}, we define the U(1)-graded partition function 
\begin{align}
\begin{aligned}
\label{eq:ell_tr}
    Z_{\mathrm{EG}}(\tau,z\,;\Lambda(C),\chi) = \Tr_{\CH_\mathrm{R}} \left[ e^{2\pi\i \,(z+\frac{1}{2})\, J_0}\,q^{L_0-\frac{n}{24}} \right]\,,
\end{aligned}
\end{align}
where $J_n$ ($n\in\BZ$) is the mode expansion of the U(1) current $J(z)$ and $q = e^{2\pi\i\tau}$.
The conformal weight and U(1) charge of each state in the Ramond sector are given by \eqref{eq:qn_r}. Then, the partition function can be written as 
\begin{align} \label{eq:elliptic_from_lattice}
\begin{aligned}
    Z_{\mathrm{EG}}(\tau,z\,;\Lambda(C),\chi) =\frac{1}{\eta(\tau)^n}\,\sum_{\lambda\,\in\,\Lambda(C)} e^{2\pi\i\,(z+\frac{1}{2})\,\chi\cdot (\lambda+\frac{\chi}{2})}\,q^{\frac{1}{2}(\lambda+\frac{\chi}{2})^2}\,,
\end{aligned}
\end{align}
where $\eta(\tau)$ is the Dedekind eta function.
For an $\CN=2$ superconformal theory, one can expect that the graded partition function gives the elliptic genus and becomes a weak Jacobi form of weight 0 and index $n/6$.
The following proposition guarantees that the partition function exhibits the expected modular transformation for $n\in12\BZ$.

\begin{proposition}
Let $\Lambda(C)$ be an odd self-dual lattice of rank $n\in12\BZ$ and $\chi\in\CX_{n/3}$. 
Then, the partition function $Z_{\mathrm{EG}}(\tau,z\,;\Lambda(C),\chi)$ 
satisfies
\begin{align}
\label{eq:elliptic_md}
    Z_{\mathrm{EG}}\left(\frac{a\tau+b}{c\tau+d},\frac{z}{c\tau+d}\,;\Lambda(C),\chi\right) &= e^{2\pi\i m \frac{cz^2}{c\tau+d}}\,Z_{\mathrm{EG}}(\tau,z\,;\Lambda(C),\chi)\,,\\
    Z_{\mathrm{EG}}(\tau,z+l\tau+l'\,;\Lambda(C),\chi) &= e^{-2\pi\i m(l^2 \tau+ 2lz)}\,Z_{\mathrm{EG}}(\tau,z\,;\Lambda(C),\chi)\,,
\label{eq:elliptic_sf}
\end{align}
where $\begin{psmallmatrix}
    a & b \\ c& d
\end{psmallmatrix}\in \mathrm{SL}(2,\BZ)$ and $l,l'\in\BZ$, and $m=n/6$.
\label{prop:ell_trans}
\end{proposition}

\begin{proof}
    The modular transformation by $\mathrm{SL}(2,\BZ)$ is generated by the modular $T$ transformation $\tau\to\tau+1$ and $S$ transformation $\tau\to-1/\tau$. First, we consider the $T$ transformation
    \begin{align}
        \begin{aligned}
            Z_{\mathrm{EG}}(\tau+1,z\,;\Lambda(C),\chi) = \frac{1}{\eta(\tau+1)^n} \sum_{\lambda\,\in\,\Lambda(C)} e^{\pi\i(\lambda+\frac{\chi}{2})^2} e^{2\pi\i\,(z+\frac{1}{2})\,\chi\cdot (\lambda+\frac{\chi}{2})}\,q^{\frac{1}{2}(\lambda+\frac{\chi}{2})^2}\,,
        \end{aligned}
    \end{align}
    where we used $q\to e^{2\pi\i} q$ under the $T$ transformation. By definition of a characteristic vector $\chi$, the phase factor $e^{\pi\i(\lambda+\frac{\chi}{2})^2}$ becomes
    \begin{align}
        e^{\pi\i(\lambda+\frac{\chi}{2})^2} = e^{\pi\i(\lambda\cdot\lambda+\lambda\cdot\chi)}\cdot e^{\frac{\pi\i}{4}\chi\cdot\chi} = e^{\frac{\pi\i n}{12}}.
    \end{align}
    This exactly cancels an overall phase coming from the $T$ transformation of $\eta(\tau)$. Then, we conclude the invariance of $Z_{\mathrm{EG}}(\tau,z\,;\Lambda(C),\chi)$ under $\tau\to\tau+1$, which agrees with \eqref{eq:elliptic_md}.

    Second, we consider the modular $S$ transformation by fixing the parameters $a=d=0$, $b=-1$, $c=1$. Then it implements $\tau\to-1/\tau$ and $z\to z/\tau$:
    \begin{align}
        Z_{\mathrm{EG}}(-1/\tau,z/\tau\,;\Lambda(C),\chi) = \frac{1}{\eta(-1/\tau)^n} \sum_{\lambda\,\in\,\Lambda(C)}\,e^{2\pi\i\,(\frac{z}{\tau}+\frac{1}{2}) \,\chi\cdot(\lambda+\frac{\chi}{2})}\, e^{-\frac{\pi\i}{\tau}(\lambda+\frac{\chi}{2})^2}\,.
    \end{align}
    Since the Fourier transform is given by
    \begin{align}
    \begin{aligned}
        \hat{f}(\omega) &= \int_{\BR^n}\d^n\lambda \; e^{-2\pi\i\lambda\cdot\omega}\,e^{2\pi\i\,(\frac{z}{\tau}+\frac{1}{2}) \,\chi\cdot(\lambda+\frac{\chi}{2})}\, e^{-\frac{\pi\i}{\tau}(\lambda+\frac{\chi}{2})^2}\,,\\
        &= (-\i\tau)^{\frac{n}{2}}\,e^{2\pi\i m\frac{z^2}{\tau}}\,q^{\frac{1}{2}(\omega-\frac{\chi}{2})^2}\,e^{2\pi\i\,(z+\frac{1}{2})\,\chi\cdot(-\omega+\frac{\chi}{2})}\,,
    \end{aligned}
    \end{align}
    the Poisson summation formula tells us
    \begin{align}
    \begin{aligned}
        Z_{\mathrm{EG}}(-1/\tau,z/\tau\,;\Lambda(C),\chi) &= \frac{1}{\eta(-1/\tau)^n}\,\sum_{\omega\,\in\,\Lambda(C)^*} \,\hat{f}(\omega)\,,\\
        &= \frac{e^{2\pi\i m\frac{z^2}{\tau}}}{\eta(\tau)^n} \sum_{\omega\,\in\,\Lambda(C)}\,q^{\frac{1}{2}(\omega+\frac{\chi}{2})^2}\,e^{2\pi\i\,(z+\frac{1}{2})\,\chi\cdot(\omega+\frac{\chi}{2})}\,,
    \end{aligned}
    \end{align}
    where we used $\chi\cdot\chi=n/3\in4\BZ$, the modular $S$ transformation $\eta(-1/\tau) = \sqrt{-\i\tau}\,\eta(\tau)$ and self-duality $\Lambda(C)^* = \Lambda(C)$, and we changed the variable $\omega\to-\omega$. Hence, we conclude the modular $S$ transformation \eqref{eq:elliptic_md}.

    Finally, we consider the transformation \eqref{eq:elliptic_sf} under the spectral flow:
    \begin{align}
    \begin{aligned}
        e^{2\pi\i m(l^2\tau+2lz)}\,Z_{\mathrm{EG}}(\tau,z+l\tau+l') &= \frac{1}{\eta(\tau)^n} \sum_{\lambda\,\in\,\Lambda(C)}e^{2\pi\i \,(z+\frac{1}{2})\,\chi\cdot\{(\lambda+l\chi)+\frac{\chi}{2}\}} \,q^{\frac{1}{2}\{(\lambda+l\chi)+\frac{\chi}{2}\}^2}\,,\\
        &= \frac{1}{\eta(\tau)^n}\sum_{\lambda\,\in\,\Lambda(C)} e^{2\pi\i(z+\frac{1}{2})\chi\cdot(\lambda+\frac{\chi}{2})}\,q^{\frac{1}{2}(\lambda+\frac{\chi}{2})^2}\,,
    \end{aligned}
    \end{align}
    where we used $\chi\cdot\chi=n/3\in 2\BZ$ and changed the variable $\lambda\to\lambda-l\chi\in\Lambda(C)$ ($l\in\BZ$). (The above transformation suggests that the shift $z\to z+l\tau$ can be understood as the shift of $\Lambda(C)$ by $l\chi$ in terms of a lattice.)
\end{proof}

Since the above proposition only ensures the modular transformation of $Z_{\mathrm{EG}}(\tau,z\,;\Lambda(C),\chi)$, we need to require an additional condition to ensure that it is a weak Jacobi form.
In the presence of supersymmetry, any operator in the Ramond sector must satisfy $h\geq n/24$, which is a unitarity bound from the superconformal algebra.
Then, from \eqref{eq:ell_tr}, the U(1)-graded partition function can be expanded as
\begin{align}
\label{eq:exp_weak}
    Z_{\mathrm{EG}}(\tau,z\,;\Lambda(C),\chi) = \sum_{n\,\geq\,0\,,\,l\,\in\,\BZ} c(n,l)\,q^n y^l\,.
\end{align}
In this case, by combining the transformation law, we can conclude that $Z_{\mathrm{EG}}(\tau,z\,;\Lambda(C),\chi)$ is a weak Jacobi form.

For fermionic code CFTs, the unitarity bound $h\geq n/24$ in the Ramond sector can be rephrased as in the following proposition.

\begin{proposition}
    Let $\Lambda(C)$ be an odd self-dual lattice.
    Then, the Ramond sector satisfies $h\geq n/24$ if and only if any characteristic vector $\chi\in2S(\Lambda(C))$ satisfies $\chi\cdot\chi\geq n/3$.
\label{prop:unitarity_bound}
\end{proposition}

\begin{proof}
    Let us rewrite the bound $h\geq n/24$ in the Ramond sector. 
    We choose an element $\chi\in 2S(\Lambda(C))$.
    Then, from \eqref{eq:qn_r}, the unitarity bound becomes
    \begin{align}
        \frac{1}{2}\left(\lambda+\frac{\chi}{2}\right)^2\geq \frac{n}{24}\,.
    \end{align}
    Using a characteristic vector $\chi\in2S(\Lambda(C))$, any characteristic vector can be represented by $\chi' = \chi + 2\lambda$ ($\lambda\in\Lambda(C)$).
    Therefore, the above inequality can be written as
    \begin{align}
        \frac{1}{2}\left(\frac{\chi'}{2}\right)^2\geq \frac{n}{24} \quad\Leftrightarrow\quad \chi'\,^2 \geq \frac{n}{3}\,.
    \end{align}
    Therefore, the unitarity bound in the R sector is equivalent to the bound $\chi\cdot\chi\geq n/3$ for any characteristic vector.
\end{proof}

Therefore, if one assumes that the minimum norm of characteristic vectors is $n/3$, then the U(1)-graded partition function $Z_{\mathrm{EG}}(\tau,z\,;\Lambda(C),\chi)$ has a Fourier expansion~\eqref{eq:exp_weak}.
Combining Proposition~\ref{prop:ell_trans}, we arrive at the following theorem.

\begin{theorem}
    Let $\Lambda(C)$ be an odd self-dual lattice of rank $n\in12\BZ$ and the minimum norm of characteristic vectors in $\Lambda(C)$ is $n/3$.
    For an element $\chi\in\CX_{n/3}$, $Z_{\mathrm{EG}}(\tau,z\,;\Lambda(C),\chi)$ is a weak Jacobi form of weight 0 and index $m = n/6$.
    \label{theorem:weak}
\end{theorem}

Under the assumption in Theorem~\ref{theorem:weak}, the partition function $Z_{\mathrm{EG}}(\tau,z\,;\Lambda(C),\chi)$ becomes a weak Jacobi form of weight 0 and index $n/6$.
This is the characterization of the elliptic genus, so it is expected to return the elliptic genus.
In section \ref{sec:extremal}, we exploit the method described above to reproduce extremal elliptic genera, which gives further evidence of the validity of our computation of elliptic genera.

Note that the U(1)-graded partition function depends on the choice of a characteristic vector $\chi\in\CX_{n/3}$.
Generally, there can be several inequivalent choices for $\chi\in\CX_{n/3}$, which give different weak Jacobi forms $Z_{\mathrm{EG}}(\tau,z\,;\Lambda(C),\chi)$.
This can reflect that there can exist several ways of realizing the $\CN=2$ superconformal symmetry depending on the choice of the U(1) current.
In what follows, we separately discuss how to choose characteristic vectors in fermionic code CFTs for an odd prime $p$ and the binary case $(p=2)$.

\paragraph{Characteristic vector for odd prime $p$}

For an odd prime $p$, the Construction A lattice is obtained from a self-dual code $C\subset\BF_p^n$.
Then, a characteristic vector of $\Lambda(C)$ is given by the following proposition, which is an improved statement of Proposition~3.7 in~\cite{Kawabata:2023nlt}.

\begin{proposition} 
Let $p$ be an odd prime and $C\subset\BF_p^n$ a self-dual code. Then $\chi\in\BR^n$ is a characteristic vector of the Construction A lattice $\Lambda(C)$ if and only if $\chi=m/\sqrt{p}\in\Lambda(C)$ where each component $m_i$ is an odd integer for $i\in\{1,2,\cdots,n\}$.
\label{prop:characteristic_odd}
\end{proposition}

\begin{proof}
We explicitly prove the necessary condition (``only if"). From the construction of $\Lambda(C)$, $\lambda\in\Lambda(C)$ can be written as $\lambda=\frac{1}{\sqrt{p}}s,\, s\in\mathbb{Z}^n$. Since $\Lambda(C)$ is self-dual and $p$ is odd, $\lambda\cdot\lambda$ is an integer and
\begin{equation}
    \lambda\cdot\lambda \equiv p\,\lambda\cdot\lambda = s\cdot s \equiv \sum_{i=1}^n s_i \mod 2\,.
\end{equation}
On the other hand, if $\chi$ is a characteristic vector, $\chi\cdot\lambda$ must be an integer, which means that $\chi\in\Lambda(C)$. Then $\chi$ can be written as $\chi=m/\sqrt{p}\,,\, m\in\mathbb{Z}^n$ and
\begin{equation}
    \chi\cdot\lambda \equiv p\,\chi\cdot\lambda = \sum_{i=1}^n m_i\,s_i \mod 2\,.
\end{equation}
Since any $\Lambda(C)$ contains vectors with one non-zero element such as $\lambda=\frac{1}{\sqrt{p}}(p,0,\cdots,0)$, the definition of characteristic vector:
\begin{equation}
    \forall \lambda\in\Lambda(C)\,,\ \lambda\cdot\lambda = \chi\cdot\lambda \mod 2
\end{equation}
is satisfied if and only if all $m_i$ are odd.
\end{proof}

Table \ref{tab:chara_vecs_odd} shows characteristic vectors $\chi=\frac{1}{\sqrt{p}}(c+pm)$ with $\chi\cdot\chi = n/3$. For $p>3$ (and $p=2$ as we will see shortly), the number of types increases as $n$ becomes larger.
For $p=3$, there is only one type of vector at any $n$.
In particular, if $C$ contains the codeword $\mathbf{1}_n\in C$, we can take the characterstic vector $\chi = \mathbf{1}_n/\sqrt{3}$. Then, the U(1)-graded partition function can be written in terms of the complete weight enumerator of $C$ as
\begin{align} \label{eq:eg_by_cwe_p3}
\begin{aligned}
    Z_{\mathrm{EG}}(\tau,z;\Lambda(C),\chi) &= \frac{1}{\eta(\tau)^n}\, \sum_{\lambda\,\in\,\Lambda(C)}\, e^{2\pi\i (z+\frac{1}{2}) \chi\cdot(\lambda+\frac{\chi}{2})} \,q^{\frac{1}{2}(\lambda+\frac{\chi}{2})^2}\,,\\
    &= \frac{1}{\eta(\tau)^n} W_C(\{f_a(\tau,z)\})\,,
\end{aligned}
\end{align}
Here, the complete weight enumerator of $C\subset\BF_3^n$ is defined by
\begin{align}
    W_C(\{x_a\}) = \sum_{c\,\in \,C}\, \prod_{a\,\in\,\BF_3} \,x_a^{\mathrm{wt}_a(c)}\,,
\end{align}
where $\mathrm{wt}_a(c) = |\{i\mid c_i = a\}|$, and
\begin{align}
    f_a(\tau,z) = \sum_{m\,\in\,\BZ} e^{2\pi\i (z+\frac{1}{2})(m+\frac{a}{3}+\frac{1}{6})} q^{\frac{3}{2}(m+\frac{a}{3}+\frac{1}{6})^2}\,,\qquad a\in\BF_3\,.
\end{align}

\clearpage
\renewcommand{\arraystretch}{1.3}

\begin{table}[tbp]
\caption{The characteristic vectors $\chi\in 2S(\Lambda(C)) \subset \BR^n$ with $\chi\cdot\chi = n/3$ and the corresponding vectors in $\BF_p^n$ for odd prime $p=3,5$. The meaning of ``corresponding" is that there exists $m\in\BZ^n$ such that $\chi=\frac{1}{\sqrt{p}}(c+pm)$. For example, $([\pm 1 \times 2,\ 0 \times 10])$ refers to a vector such that two components are $1$ or $-1$ and ten components are $0$.}
\label{tab:chara_vecs_odd}
\centering
\begin{minipage}{\textwidth}
\subcaption{$p=3$}
\begin{equation*}
\begin{array}{c@{\hspace{2em}}c@{\hspace{2em}}c}
    n & \chi\in \CX_{n/3} & c\in C \\ \hline
    12 & \frac{1}{\sqrt{3}}([\pm 1 \times 12]) & ([(1 \text{ or } 2) \times 12]) \\
    24 & \frac{1}{\sqrt{3}}([\pm 1 \times 24]) & ([(1 \text{ or } 2) \times 24]) \\
    36 & \frac{1}{\sqrt{3}}([\pm 1 \times 36]) & ([(1 \text{ or } 2) \times 36]) \\
\end{array}
\end{equation*} \\
\end{minipage}
\begin{minipage}{\textwidth}
\subcaption{$p=5$}
\begin{equation*}
\begin{array}{c@{\hspace{2em}}c@{\hspace{2em}}c}
    n & \chi\in \CX_{n/3} & c\in C \\ \hline
    12 & \frac{1}{\sqrt{5}}([\pm 1 \times 11,\ \pm 3 \times 1]) & ([(1 \text{ or } 4) \times 11,\ (2 \text{ or } 3) \times 1]) \\
    24 & \frac{1}{\sqrt{5}}([\pm 1 \times 22,\ \pm 3 \times 2]) & ([(1 \text{ or } 4) \times 22,\ (2 \text{ or } 3) \times 2]) \\
    36 & \frac{1}{\sqrt{5}}([\pm 1 \times 33,\ \pm 3 \times 3]) & ([(1 \text{ or } 4) \times 33,\ (2 \text{ or } 3) \times 3]) \\
    & \frac{1}{\sqrt{5}}([\pm 1 \times 35,\ \pm 5 \times 1]) & ([0 \times 1,\ (1 \text{ or } 4) \times 35]) 
\end{array}
\end{equation*}
\end{minipage}
\end{table} 

\bigskip

\begin{table}[tbp]
\caption{The characteristic vectors $\chi\in 2S(\Lambda(C)) \subset \BR^n$ with $\chi\cdot\chi = n/3$ and the corresponding vectors in $\BF_2^n$. The meaning of ``corresponding" is that there exists $m\in\BZ^n$ such that $ \chi=\sqrt{2}(s+2m)$.}
\label{tab:chara_vecs_p2}
\begin{minipage}{\textwidth}
\begin{equation*}
\begin{array}{c@{\hspace{2em}}c@{\hspace{2em}}c}
    n & \chi\in \CX_{n/3} & s\in S(C) \\ \hline
    12 & \sqrt{2}([\pm 1 \times 2,\ 0 \times 10]) & ([1 \times 2,\ 0 \times 10]) \\
    24 & \sqrt{2}([\pm 1 \times 4,\ 0 \times 20]) & ([1 \times 4,\ 0 \times 20]) \\
    36 & \sqrt{2}([\pm 2 \times 1, \pm 1 \times 2,\ 0 \times 33]) & ([1 \times 2,\ 0 \times 34]) \\
    & \sqrt{2}([\pm 1 \times 6,\ 0 \times 30]) & ([1 \times 6,\ 0 \times 30]) \\
\end{array}
\end{equation*}
\end{minipage}
\end{table}
\renewcommand{\arraystretch}{0.9}

\clearpage

\paragraph{Characteristic vector for $p=2$}
For a binary singly-even self-dual code $C\subset\BF_2^n$, one can define the shadow of the code, from which we can obtain the shadow of the lattice $\Lambda(C)$.

We divide $C$ into two subsets $C_0$ and $C_2$ consisting of doubly-even and singly-even codewords of $C$, respectively:
\begin{align}
\begin{aligned}
    C_0 &= \{ c\in C \mid c\cdot c\in4\BZ \}\,, \\
    C_2 &= \{ c\in C \mid c\cdot c\in4\BZ+2 \}\,,
\end{aligned}
\end{align}
where $\cdot$ is the Euclidean inner product: $c\cdot c' = \sum_{i} c_i \,c_i'\in\BR$.
Note that $C_0$ is a doubly-even self-orthogonal code and $C_0\subset C\subset C_0^\perp$. The shadow $S(C)$ of the code $C$ is defined by (\cite{conway1990new_code})
\begin{align}
    S(C) = C_0^\perp \setminus C\,.
\end{align}
We can relate the shadow of the code to characteristic vectors of the lattice by the following proposition.

\begin{proposition}[{\cite{elkies1999lattices}}]
    Let $C\subset \BF_2^n$ be a singly-even self-dual code. Then $\chi \in\BR^n$ is a characteristic vector of $\Lambda(C)$ if and only if $\chi= \sqrt{2}\,(s+2m)$ where $s\in S(C)$ and $m\in \BZ^n$. 
\end{proposition}

Table~\ref{tab:chara_vecs_p2} shows characteristic vectors $\chi=\sqrt{2}(s+2m)$ with $\chi\cdot\chi = n/3$. Note that $([1 \times 2,\ 0 \times 34])\in C$ at $n=36$ also corresponds to the characteristic vector $\chi'=\sqrt{2}([\pm1 \times 2,\ 0 \times 34])$, which does not satisfy the unitarity bound from Proposition~\ref{prop:unitarity_bound}.

In the rest of this subsection, we show a practical way to find shadows for specific codes.
Here, we use the inclusion map $\iota: \BF_2\to \{0,1\}\subset\BZ$ to avoid some confusion.

\begin{proposition}
Let $C\subset\BF_2^n$ be a singly-even self-dual code that has an $\frac{n}{2} \times n$ generator matrix $G$ such that the left side is the identity matrix $I_{n/2}$. Then,
\begin{equation}
    s = \left( P(G_1),\dots,P(G_{n/2}),0,\dots,0 \right) \in \BF_2^n
\end{equation}
is an element of the shadow $S(C)$ where
\begin{align}
    G_i \in C &: \text{ the $i$-th row of } G \,, \\
    P(c) &= \begin{dcases} 0 & (c\cdot c\in4\BZ)\,, \\ 
    1 & (c\cdot c\in4\BZ+2)\,. \end{dcases} 
\end{align}
\end{proposition}

\begin{proof}
For any codeword $c=x\,G\in C$, $x\in\BF_2^{n/2}$, there exists $m\in\BZ^n$ such that $\iota(c)=\iota(x)\iota(G)+2m$. Then, via mod 4 computation, we obtain
\begin{equation}
\begin{aligned}
    c\cdot c &= (\iota(x)\iota(G)+2m)^2 \\
    &\equiv (\iota(x)\iota(G))^2 \\
    &= \sum_{i} \iota(x_i)^2 \,\iota(G_i)\cdot \iota(G_i) + 2 \sum_{i\neq j} \iota(x_i)\iota(x_j) \,\iota(G_i)\cdot \iota(G_j) \\
    &\equiv 2 \sum_{i} \iota(x_i\, P(G_i)) \,.
\end{aligned}
\end{equation}
In the last line, we used $c\cdot c\equiv 2P(c)$ and $c\cdot c'\in2\BZ$ for $c,c'\in C$ from the self-duality of $C$. On the other hand, since the left part of $G$ is the identity matrix $I_{n/2}$,
\begin{equation}
    c\cdot s = x\,G \cdot \left( P(G_1),\dots,P(G_{n/2}),0,\dots,0 \right) = \sum_i \iota(x_i \,P(G_i))\,,\mod 2 \,.
\end{equation}
Thus, the equivalences
\begin{align*}
    c\cdot s = 0 \;\;\mathrm{mod}\;2 \;&\Leftrightarrow\; c\in C: \text{doubly-even}\\
    c\cdot s = 1 \;\;\mathrm{mod}\;2 \;&\Leftrightarrow\; c\in C: \text{singly-even}
\end{align*}
hold, which means that $s\in S(C)$.
\end{proof}

From the proposition, we can read off an element of the shadow $S(C)$ from the generator matrix $G$. Since the shadow of $C$ is a coset of $C$ in $\BF_2^n$: $s+C=S(C)$, the whole set of $S(C)$ can be easily obtained once we have an element $s\in S(C)$.

\subsection{Spectral flow and chiral ring}
\label{ss:spectral_flow}

We have shown that the U(1)-graded partition function becomes a weak Jacobi form under some assumptions.
This strongly supports the existence of $\CN=2$ supersymmetry in the corresponding fermionic code CFT.
In this subsection, we consider the $\CN=2$ structure and, more concretely, aim to identify the spectral flow and chiral rings in fermionic code CFTs.

Let us consider a fermionic CFT constructed from an odd self-dual lattice $\Lambda(C)$ where the minimum norm of characteristic vectors is $n/3$ and specify the U(1) current by \eqref{eq:cand_R_current} where $\chi\in\CX_{n/3}$.
For the fermionic code CFT, a state with charge $Q = \chi\cdot\kappa$ can be represented by
\begin{align}
    \begin{aligned}
        V_{\kappa}(z) = \,:e^{\i\kappa \cdot X(z)}:\,,\qquad \kappa \in
        \begin{dcases}
            \Lambda(C) & (\text{NS sector})\,,\\
            S(\Lambda(C)) & (\text{R sector})\,,
        \end{dcases}
    \end{aligned}
\end{align}
up to neutral operators.
The spectral flow by $\theta$ can be obtained by shifting the momentum $\kappa$ by $\theta\chi$, i.e. $\kappa\to \kappa-\theta\chi$.
In fact, under the flow, a state $V_\kappa(z)$ with $h = \kappa^2/2$ and $Q = \kappa\cdot\chi$ transforms into $V_{\kappa-\theta\chi}(z) = \,:\exp[\i\,(\kappa-\theta\chi)\cdot X(z)]:$ with
\begin{align}
    \begin{aligned}
        h' &= \frac{1}{2}(\kappa-\theta\chi)^2 = h - \theta Q + \frac{n}{6}\theta^2\,,\\
        Q' &= \chi\cdot(\kappa-\theta\chi) = Q -\frac{n}{3}\theta\,,
    \end{aligned}
\end{align}
where we used $\chi\cdot\chi = n/3$.
This reproduces the transformation rule \eqref{eq:spct_qn} of quantum numbers under the spectral flow.
Therefore, we identify the spectral flow operator as
\begin{align}
    \mathscr{U}_\theta = \,:\exp\left(-\i\,\theta\,\chi\cdot X(z)\right):\,.
\end{align}
This is also consistent with the observation in the proof of Proposition \ref{prop:ell_trans}: the spectral flow can be understood as the shift of $\Lambda(C)$ by $\chi$ in terms of a lattice.

From Proposition \ref{prop:unitarity_bound}, the assumption that the minimum norm of characteristic vectors is $n/3$ implies that any operator in the R sector satisfies $h\geq n/24$. 
In this case, we may identify the Ramond ground states as
\begin{align} \label{eq:R_ground}
    V_{\chi'/2}(z) =\,:\exp(\i\,\chi'\cdot X(z)/2):\,,
\end{align}
where $\chi'\in \CX_{n/3}$. The vertex operators have the expected conformal weight $h = n/24$.

Furthermore, there is another unitarity bound \eqref{eq:uni_bound_NS} in the NS sector.
Interestingly, the fermionic code CFT also satisfies the unitarity bound $h\geq Q/2$ in the NS sector.
To see this, suppose a characteristic vector $\chi\in 2S(\Lambda(C))$ with the minimum norm.
Then, any characteristic vector $\chi'\in2S(\Lambda(C))$ is written in the form $\chi' = \chi-2\lambda$ $(\lambda\in\Lambda(C))$.
Since $\chi\in 2S(\Lambda(C))$ is minimum: $\chi'\cdot\chi'\geq\chi\cdot\chi$, we obtain
\begin{align}
    \lambda\cdot\lambda \geq \lambda\cdot \chi\,,
\end{align}
which implies $h\geq Q/2$ for a vertex operator $V_\lambda(z)$. The bosonic oscillators increase the $L_0$ eigenvalue and do not change the U(1) charge, so any operator in the NS sector satisfies the bound.

As explained in section~\ref{ss:N=2}, the spectral flow by $\theta=-1/2$ flows the Ramond ground states into chiral primary fields. Now we obtain
\begin{align}
    V_{(\chi+\chi')/2}(z) =\,:\exp[\i\,(\chi +\chi')\cdot X(z)/2]:\,.
\end{align}
The conformal weight $h$ and U(1) charge $Q$ of the states are
\begin{align}
    \begin{aligned}
        h &=\frac{1}{2}\left(\frac{\chi+\chi'}{2}\right)^2 = \chi\cdot(\chi +\chi')/4\,,\\
        Q &= \chi\cdot(\chi+\chi')/2\,,
    \end{aligned}
\end{align}
which implies $h = Q/2$, reproducing the known relation between quantum numbers for chiral primary fields. Also, the conformal weight $h$ satisfies
\begin{align}
    h = \chi\cdot(\chi+\chi')/4 \leq \chi\cdot(\chi+\chi)/4 = n/6\,.
\end{align}
This also agrees with the property of chiral primary fields as shown in~\eqref{eq:bound_chiralpr}.

Let us consider the operator product expansion of two chiral primary fields
\begin{align}
    V_{(\chi+\chi')/2}(z)\,V_{(\chi+\chi'')/2}(w) \sim (z-w)^{\frac{\chi+\chi'}{2}\cdot\frac{\chi+\chi''}{2}} \,V_{\chi +(\chi'+\chi'')/2}(w)\,.
\end{align}
Using the fact that any operator satisfies $h\geq Q/2$ and chiral primary fields saturate it, in this case, we obtain
\begin{align}
    \frac{\chi+\chi'}{2}\cdot \frac{\chi+\chi''}{2} \geq0\,.
\end{align}
Therefore, if one takes the limit $z\to w$, the operator product does not vanish if and only if $(\chi+\chi')\cdot(\chi+\chi'') = 0$.
Hence the chiral ring can be written as
\begin{align}
    \mathscr{R} = \left\{ r_a \;\middle|\; a = \frac{\chi+\chi'}{2}\,,\; \chi'\in \CX_{n/3}\right\}\,.
\end{align}
Here, the product is $r_a\cdot r_b = r_{a+b}$ subject to $r_a\cdot r_b = 0$ if $a\cdot b > 0$ where $a = (\chi+\chi')/2$ and $b = (\chi+\chi'')/2$.

\paragraph{In the ternary case}

For a ternary code $(p=3)$, the above discussion can be further simplified.
In the rest of this section, we move on to the convenient frame $\BF_3 = \{-1,0,1\}$ using the isomorphism $\varphi:\{0,1,2\}\to\{-1,0,1\}$ where
\begin{align}
    \varphi: 0\mapsto 0\,,\quad 1\mapsto 1\,,\quad 2\mapsto-1\,.
\end{align}
Then, a characteristic vector $\chi\in\CX_{n/3}$ takes the form $\chi = \frac{1}{\sqrt{3}}([(\pm1)\times n])$, and the corresponding codeword is $c = ([(\pm1)\times n])\in C$.
This type of codeword is called maximal and we denote the set of maximal codewords by $C_{\max}$. Note that a self-dual ternary code contains maximal codewords only when $n\in12\BZ$.
Then, a characteristic vector $\chi\in\CX_{n/3}$ can be written as 
\begin{align}
    \chi = c/\sqrt{3}\,,\qquad c\in C_{\max}\,.
\end{align}
Choosing a characteristic vector by $\chi = c/\sqrt{3}$ $(c\in C_{\max})$, the spectral flow operator is
\begin{align}
    \EU_\theta = \,:\exp(-\i\,\theta\, c\cdot X(z)/\sqrt{3}):\,,
\end{align}
where $\theta$ is a twist parameter.
The Ramond ground states can be identified as
\begin{align}
    V_{\chi'/2}(z) = \,:e^{\frac{\i}{2\sqrt{3}}\,c'\cdot X(z)}:\,,\qquad c'\in C_{\max}\,,
\end{align}
where we used $\chi' = c'/\sqrt{3}$.
Correspondingly, the chiral primary fields in the NS sector are
\begin{align}
    V_{(\chi+\chi')/2}(z) = \,:e^{\frac{\i}{2\sqrt{3}}\,(c+c')\cdot X(z)}:\,,
\end{align}
which gives the chiral ring
\begin{align}
    \mathscr{R} = \left\{r_{(c+c')/2}\;\middle|\; c'\in C_{\max}\right\}\,.
\end{align}
The U(1) character valued degeneracy of the Ramond ground states is
\begin{align}
    \left.\Tr_\mathrm{R}\left[t^{J_0}\right]\right|_{G_0^\pm = 0} = \sum_{\chi'\in\CX_{n/3}} t^{\chi\cdot\chi'/2} = \sum_{c'\in C_{\max}} t^{c\cdot c'/6}\,,
\end{align}
where we used $\chi = c/\sqrt{3}$. 
By the spectral flow, we obtain the counterpart in the NS sector called the Poincar\'{e} polynomial \cite{Lerche:1989uy}
\begin{align}
    P(t) := \Tr_{\mathrm{NS}}\left[t^{J_0}\right]|_{\mathscr{R}} = \sum_{c'\in C_{\max}} t^{c\cdot(c+ c')/6} = t^{n/6} \sum_{c'\in C_{\max}}  t^{c\cdot c'/6}\,,
\end{align}
where we used $c\cdot c  = n$.
These polynomials are completely determined by the set of maximal codewords.
The Poincar\'{e} polynomial is known to be closely related to the cohomology of the target manifold~\cite{Lerche:1989uy}.

\section{Extremal $\CN=2$ elliptic genera}
\label{sec:extremal}

In this section, we construct extremal $\CN=2$ superconformal theories with the central charge $n=12$ and $24$.
Originally, bosonic and $\CN=1$ extremal CFTs are introduced to understand the holographic duality of the pure AdS$_3$ gravity~\cite{Witten:2007kt}.\footnote{There have been numerous works on the relation between extremal CFTs and pure gravity. See, e.g.,~\cite{Gaiotto:2007xh,Gaberdiel:2007ve,Yin:2007gv,Yin:2007at,Maloney:2007ud,Gaiotto:2008jt,Li:2008dq,Gaberdiel:2008pr,Gaberdiel:2010jf,Benjamin:2016aww,Bae:2016yna,Ferrari:2017kbp} for the developments.}
For $\CN=2$ theories, extremal CFTs are defined using extremal elliptic genera~\cite{Gaberdiel:2008xb}.
The extremal $\CN=2$ SCFTs with the central charge $n=12$ and $24$ have been explicitly constructed in \cite{Cheng:2014owa,Benjamin:2015ria}.
We will give an alternative systematic construction of the extremal elliptic genera from linear codes.
In this and the following sections, we use the graded character defined in Appendix \ref{app:N2_character}. Since we only consider the partition function in the R sector, we abbreviate $\mathrm{Ch}_{l;h,Q}^{(\mathrm{R})}$ as $\mathrm{Ch}_{l;h,Q}$.

\subsection{central charge $12$}
\label{ss:ext_m=2}

In this subsection, we employ our method to obtain the extremal elliptic genus with the central charge $n=12$ (correspondingly, $m=2$)~\cite{Cheng:2014owa}.
As in \eqref{eq:extremal_m2}, the $q^0$ term in the extremal elliptic genus at $m=2$ is 
\begin{align}
    y^2 + 22 + y^{-2}\,.
\end{align}
From the discussion around \eqref{eq:R_ground}, the corresponding lattice $\Lambda(C)$ must have 24 orthogonal characteristic vectors in $\CX_{n/3}$, i.e., $|\CX_{n/3}|=24$ and $\chi\cdot\chi'=0$ for $\chi,\chi'\in\CX_{n/3}, \chi\neq\pm\chi'$.

In what follows, we search for classical codes that give the extremal elliptic genus by imposing constraints from the orthogonality of characteristic vectors.
Here, we consider classical codes over $\BF_p$ where $p=2,5$.

When $p=2$, from Table \ref{tab:chara_vecs_p2}, the shadow $S(C)$ of the code $C\subset\BF_2^{12}$ must have 6 vectors with two $1$ and ten $0$. In addition, from the orthogonality, all `1's must be in different positions such as
\begin{equation}
    (1,1,0^{10})\,,\ (0^2,1,1,0^8)\,,\ (0^4,1,1,0^6)\,,\ (0^6,1,1,0^4)\,,\ (0^8,1,1,0^2)\,,\ (0^{10},1,1) \,.
\end{equation}
In this case, the code $C$ must include the codewords as $(1,1,1,1,0^8)$ since $s-s'\in C$ for any $s,s'\in S(C)$ and $t\in C_2\subset C$ must be in the form like $(1,0,1,0,1,0,1,0,1,0,1,0)$ since $t\cdot s=1$ for any $s\in S(C)$. In summary, the code must be generated by
\begin{equation}
\begin{bmatrix}
    1 & 1 & 1 & 1 & 0 & 0 & 0 & 0 & 0 & 0 & 0 & 0 \\
    1 & 1 & 0 & 0 & 1 & 1 & 0 & 0 & 0 & 0 & 0 & 0 \\
    1 & 1 & 0 & 0 & 0 & 0 & 1 & 1 & 0 & 0 & 0 & 0 \\
    1 & 1 & 0 & 0 & 0 & 0 & 0 & 0 & 1 & 1 & 0 & 0 \\
    1 & 1 & 0 & 0 & 0 & 0 & 0 & 0 & 0 & 0 & 1 & 1 \\
    1 & 0 & 1 & 0 & 1 & 0 & 1 & 0 & 1 & 0 & 1 & 0 \\
\end{bmatrix} \,,
\end{equation}
up to permutations of the columns.

When $p=5$, from Table \ref{tab:chara_vecs_odd}, the code $C\subset\BF_5^{12}$ must have 24 codewords with eleven $1$ or $4$ and one $2$ or $3$. All self-dual codes over $\BF_5$ with length $12$ are known \cite{database} and there are 3 such codes up to permutations, all satisfying the orthogonality. One of them is generated by
\begin{equation}
\begin{bmatrix}
    3 & 1 & 1 & 1 & 1 & 1 & 1 & 1 & 1 & 1 & 1 & 1 \\
    1 & 3 & 4 & 4 & 4 & 4 & 4 & 4 & 4 & 4 & 1 & 1 \\
    1 & 4 & 3 & 4 & 4 & 4 & 1 & 4 & 4 & 4 & 1 & 4 \\
    1 & 4 & 4 & 3 & 4 & 4 & 4 & 1 & 4 & 4 & 1 & 4 \\
    1 & 4 & 4 & 4 & 3 & 4 & 4 & 4 & 1 & 4 & 1 & 4 \\
    1 & 4 & 4 & 4 & 4 & 3 & 4 & 4 & 4 & 1 & 1 & 4 \\
\end{bmatrix} \,.
\end{equation}

The Construction A lattice from either code is $\mathrm{D}_{12}^+$ (see \cite{conway2013sphere,lattice_catalogue} for the names of lattices). We can explicitly check that for any characteristic vector $\chi\in\Lambda(C)$ with norm $n/3=4$, the elliptic genus is
\begin{equation}
\begin{aligned}
    Z_{\mathrm{EG}}(\tau,z; \Lambda(C),\chi) &= y^2 + 22 + y^{-2} + \CO(q) \\
    &= \frac{1}{6} \phi_{0,1}^2 + \frac{5}{6} \phi_{-2,1}^2 E_4 \\
    &= 23 \, \mathrm{Ch}_{\frac{3}{2};\frac{1}{2},0} + \mathrm{Ch}_{\frac{3}{2};\frac{1}{2},2} \\
    &\quad + \sum_{t=1}^\infty \left( A_{t,1}  (\mathrm{Ch}_{\frac{3}{2};\frac{1}{2}+t,1} + \mathrm{Ch}_{\frac{3}{2};\frac{1}{2}+t,-1}) + A_{t,2} \, \mathrm{Ch}_{\frac{3}{2};\frac{1}{2}+t,2} \right)
\end{aligned}
\end{equation}
where $A_{t,1}=\{770,13915,\dots\}$ and $A_{t,2}=\{231,5796,\dots\}$ for $t=\{1,2,\dots\}$ \cite{Harrison:2016hbq}.
This means that both codes with $p=2,5$ correspond to the CFT with the extremal elliptic genus.

In the ternary case, we can perform further analysis using the complete weight enumerator. We consider the ternary Golay code $C\subset\BF_3^{12}$ generated by
\begin{equation}
\begin{bmatrix}
    1 & 0 & 0 & 0 & 0 & 0 & 0 & 1 & 1 & 1 & 1 & 1 \\
    0 & 1 & 0 & 0 & 0 & 0 & 2 & 0 & 1 & 2 & 2 & 1 \\
    0 & 0 & 1 & 0 & 0 & 0 & 2 & 1 & 0 & 1 & 2 & 2 \\
    0 & 0 & 0 & 1 & 0 & 0 & 2 & 2 & 1 & 0 & 1 & 2 \\
    0 & 0 & 0 & 0 & 1 & 0 & 2 & 2 & 2 & 1 & 0 & 1 \\
    0 & 0 & 0 & 0 & 0 & 1 & 2 & 1 & 2 & 2 & 1 & 0 \\
\end{bmatrix} \,,
\end{equation}
whose corresponding lattice $\Lambda(C)$ is $\mathrm{D}_{12}^+$. The complete weight enumerator is
\begin{equation}
    W_C(x_0,x_1,x_2) = x_0^{12} + x_1^{12} + x_2^{12} + 22 (x_0^6x_1^6 + x_0^6x_2^6 + x_1^6x_2^6) + 220 (x_0^3x_1^3x_2^6 + x_0^3x_1^6x_2^3 + x_0^6x_1^3x_2^3) \,.
\end{equation}
The code contains the codeword $\mathbf{1}=(1,\dots,1)\in C$ and thus we can take the characteristic vector $\chi=\mathbf{1}/\sqrt{3}$. From \eqref{eq:eg_by_cwe_p3}, it can be verified that the $q^0$ term of the elliptic genus $ Z_{\mathrm{EG}}(\tau,z;\Lambda(C),\chi)$ is $y^2+22+y^{-2}$, which comes from $x_0^{12}+22x_0^6x_2^6+x_2^{12}$.

\subsection{central charge $24$}

In this subsection, we obtain the extremal elliptic genus with the central charge $n=24$ (correspondingly, $m=4$)~\cite{Benjamin:2015ria}.
As in \eqref{eq:extremal_m4}, the $q^0$ term in the extremal elliptic genus at $m=4$ is
\begin{align}
    y^4 + 46 + y^{-4}\,.
\end{align}
From the discussion around \eqref{eq:R_ground}, the corresponding lattice $\Lambda(C)$ must have 48 orthogonal characteristic vectors in $\CX_{n/3}$.

As in the previous case, we would like to find classical codes that return the extremal elliptic genus.
Below we consider classical codes over $\BF_p$ where $p=2,3$.

When $p=2$, from Table \ref{tab:chara_vecs_p2}, we have to choose the $\mathrm{U}(1)$ current such as $J(z)=\chi\cdot L(z),\, \chi=\sqrt{2}\,(1,1,1,1,0^{20})$. However, in this case, $s=\frac{1}{\sqrt{2}}(-1,1,1,1,0^{20})$ is an element of $S(\Lambda(C))$ and $\chi\cdot s=2$, thus $V_{s}(z)$ corresponds to $q^0y^2$ in the elliptic genus. The existence of this term is inconsistent with $y^4 + 46 + y^{-4}$, which means that the extremal elliptic genus cannot be constructed from classical codes over $\BF_2$.

When $p=3$, from Table \ref{tab:chara_vecs_odd}, the code $C\subset\BF_3^{24}$ must have 48 codewords consisting of only 1 and 2. In addition, from the orthogonality, if we choose $\chi\equiv\frac{1}{\sqrt{3}}c, c\in C$, then 46 codewords except for $\pm c$ must have 12 components with the same value as $c$ (both 1 or both 2) and 12 components with different values (one is 1 and the other is 2).

In this case, it is also useful to consider the NS sector. The corresponding partition function to the extremal elliptic genus is
\begin{equation}
    \mathrm{SF}_{1/2} \, Z_{\mathrm{ext}}^{(m=4)}(\tau,z-\tfrac{1}{2}) = q^{-1} + 24 + q^{\frac{1}{2}}(2048y + 2048y^{-1}) + \CO(q) \,.
\end{equation}
The terms $q^{-1}$ and $24$ come from the vacuum $\ket{0}$ and its descendants $\alpha_{-1}^i\ket{0}$, thus the minimum norm of the lattice $\Lambda(C)$ must be $3$, which corresponds to $q^{\frac{1}{2}}$. There is only one such odd self-dual lattice: the odd Leech lattice~\cite{o1944construction}.
In the ternary case, the odd Leech lattice can be constructed from the code with the minimum Hamming weight $9$. There are two such self-dual codes and one is generated by
\begin{equation}
\begin{bmatrix}
    1&0&0&0&0&0&0&0&0&0&0&0&0&2&0&1&1&1&2&0&2&1&0&1 \\
    0&1&0&0&0&0&0&0&0&0&0&0&0&1&1&1&1&2&2&0&1&0&2&0 \\
    0&0&1&0&0&0&0&0&0&0&0&0&1&0&2&0&2&1&1&1&1&0&0&2 \\
    0&0&0&1&0&0&0&0&0&0&0&0&1&1&0&0&2&2&0&1&1&1&2&0 \\
    0&0&0&0&1&0&0&0&0&0&0&0&2&2&2&0&1&0&0&1&1&1&0&1 \\
    0&0&0&0&0&1&0&0&0&0&0&0&0&0&1&2&2&2&0&1&2&0&2&2 \\
    0&0&0&0&0&0&1&0&0&0&0&0&1&0&0&1&1&2&2&2&0&1&1&0 \\
    0&0&0&0&0&0&0&1&0&0&0&0&2&2&2&1&0&2&1&2&2&1&1&1 \\
    0&0&0&0&0&0&0&0&1&0&0&0&2&0&2&2&1&1&2&1&2&1&1&2 \\
    0&0&0&0&0&0&0&0&0&1&0&0&0&2&1&1&2&1&0&1&1&2&0&0 \\
    0&0&0&0&0&0&0&0&0&0&1&0&2&1&0&0&2&1&2&0&2&0&2&2 \\
    0&0&0&0&0&0&0&0&0&0&0&1&2&2&2&1&1&1&1&0&1&2&2&2
\end{bmatrix} \,,
\end{equation}
which is modified from the top matrix for $\BF_3,\ n=24,\ d=9$ in \cite{database}.

For any characteristic vector $\chi\in\Lambda(C)$ with norm $n/3=8$ such as $\chi=\mathbf{1}/\sqrt{3}$, the elliptic genus is
\begin{equation}
\begin{aligned}
    Z_{\mathrm{EG}}(\tau,z; \Lambda(C),\chi) &= y^4 + 46 + y^{-4} + \CO(q) \\
    &= \frac{1}{432} \phi_{0,1}^4 + \frac{1}{8} \phi_{0,1}^2 \phi_{-2,1}^2 E_4 + \frac{11}{27} \phi_{0,1} \phi_{-2,1}^3 E_6 + \frac{67}{144} \phi_{-2,1}^4 E_4^2 \\
    &= 47 \, \mathrm{Ch}_{\frac{7}{2};1,0} + \mathrm{Ch}_{\frac{7}{2};1,4}  \\
    &\quad + \sum_{t=1}^\infty \left( \sum_{k=1}^3 A_{t,k} (\mathrm{Ch}_{\frac{7}{2};1+t,k} + \mathrm{Ch}_{\frac{7}{2};1+t,-k}) + A_{t,4} \, \mathrm{Ch}_{\frac{7}{2};1+t,4} \right)
\end{aligned}
\end{equation}
where $A_{t,1}=\{32890,2969208,\dots\}$, $A_{t,2}=\{14168,1659174,\dots\}$, $A_{t,3}=\{2024,485001,\dots\}$, and $A_{t,4}=\{23,61984,\dots\}$ for $t=\{1,2,\dots\}$.
Therefore, the CFT constructed from this code certainly has the extremal elliptic genus.

For $\chi=\mathbf{1}/\sqrt{3}$, from \eqref{eq:eg_by_cwe_p3}, the terms $y^4 + 46 + y^{-4}$ in the elliptic genus come from $x_0^{24}+46\,x_0^{12}x_2^{12}+x_2^{24}$ in the complete weight enumerator.

\section{Other examples}
\label{sec:near-extremal}

In the previous section, we constructed the extremal elliptic genera.
In this section, we give other examples of the elliptic genera, especially near-extremal elliptic genera, which are slightly modified in the polar regions from extremal ones as in \eqref{eq:near_extremal}.

By Theorem \ref{theorem:weak}, for the partition function $Z_{\mathrm{EG}}(\tau,z\,;\Lambda(C),\chi)$ to be a weak Jacobi form, the minimum norm of characteristic vectors in $\Lambda(C)\subset\BR^n$ must be $n/3$. From the fact that $\chi\cdot\chi=n \mod 8$ for any characteristic vector $\chi$, if $n=12,24$, it is sufficient to verify the existence of a characteristic vector with norm $n/3$, i.e., $\CX_{n/3}\neq\emptyset$. In terms of codes, this is equivalent to the existence of the codewords listed in Table \ref{tab:chara_vecs_odd} and \ref{tab:chara_vecs_p2}. If $n=36$ or greater, there may be characteristic vectors with a norm less than $n/3$. Therefore, we should choose a code $C$ whose corresponding lattice $\Lambda(C)$ does not contain such vectors. In the ternary case, the minimum norm of the characteristic vectors is always greater than or equal to $n/3$ from Proposition \ref{prop:characteristic_odd}, thus at any $n\in12\BZ$ the conditions are summarized as the existence of the maximal codewords (consisting of only $1$ and $2$), i.e., $C_{\mathrm{max}}\neq\emptyset$.

\subsection{$n=12\,,\,$ $p=2$}

In dimension $n=12$, there are only 3 odd self-dual lattices: $\mathrm{D}_{12}^+,\ \mathrm{E}_8\oplus\BZ^4$, and $\BZ^{12}$. As we have seen in section~\ref{ss:ext_m=2}, $\mathrm{D}_{12}^+$ corresponds to the CFT with the extremal elliptic genus. Since $\BZ^{12}$ is trivial, we consider the case $\mathrm{E}_8\oplus\BZ^4$.

For a code $C\subset\BF_2^{12}$ generated by
\begin{equation}
\begin{bmatrix}
    1 & 1 & 0 & 0 & 0 & 0 & 0 & 0 & 0 & 0 & 0 & 0 \\
    0 & 0 & 1 & 1 & 0 & 0 & 0 & 0 & 0 & 0 & 0 & 0 \\
    0 & 0 & 0 & 0 & 1 & 0 & 0 & 0 & 0 & 1 & 1 & 1 \\
    0 & 0 & 0 & 0 & 0 & 1 & 0 & 0 & 1 & 0 & 1 & 1 \\
    0 & 0 & 0 & 0 & 0 & 0 & 1 & 0 & 1 & 1 & 0 & 1 \\
    0 & 0 & 0 & 0 & 0 & 0 & 0 & 1 & 1 & 1 & 1 & 0 \\
\end{bmatrix} \,,
\end{equation}
the Construction A lattice $\Lambda(C)$ is $\mathrm{E}_8\oplus\BZ^4$. We can easily verify that the first four columns correspond to $\BZ^4$ and the rest to $\mathrm{E}_8$.

For any characteristic vector $\chi\in\Lambda(C)$ with norm $n/3=4$ such as
\begin{equation}
    \chi = \sqrt{2} \, (1,0,1,0,0,0,0,0,0,0,0,0) \,,
\end{equation}
the elliptic genus is
\begin{equation}
\begin{aligned}
    Z_{\mathrm{EG}}(\tau,z; \Lambda(C),\chi) &= y^2 -4y + 6 - 4y^{-1} + y^{-2} + \CO(q) \\
    &= \phi_{-2,1}^2 E_4 \\
    &= 7 \, \mathrm{Ch}_{\frac{3}{2};\frac{1}{2},0} + 4 (\mathrm{Ch}_{\frac{3}{2};\frac{1}{2},1} + \mathrm{Ch}_{\frac{3}{2};\frac{1}{2},-1}) + \mathrm{Ch}_{\frac{3}{2};\frac{1}{2},2} \\
    &\quad + \sum_{t=1}^\infty \left( A_{t,1} (\mathrm{Ch}_{\frac{3}{2};\frac{1}{2}+t,1} + \mathrm{Ch}_{\frac{3}{2};\frac{1}{2}+t,-1}) + A_{t,2} \, \mathrm{Ch}_{\frac{3}{2};\frac{1}{2}+t,2} \right)
\end{aligned}
\end{equation}
where $A_{t,1}=\{750,13875,\dots\}$ and $A_{t,2}=\{255,5868,\dots\}$ for $t=\{1,2,\dots\}$.
This is not extremal but a $\beta$-extremal elliptic genus with $\beta>1$ since it includes the term $q^0 y$. 

\subsection{$n=24\,,\,$ $p=3$}
We consider a code $C\subset\BF_3^{24}$ generated by
\begin{equation}
\begin{bmatrix}
    1 & 0 & 0 & 0 & 0 & 0 & 0 & 0 & 0 & 0 & 0 & 0 & 0 & 1 & 1 & 1 & 2 & 0 & 2 & 0 & 0 & 1 & 2 & 2 \\
    0 & 1 & 0 & 0 & 0 & 0 & 0 & 0 & 0 & 0 & 0 & 0 & 0 & 1 & 1 & 2 & 1 & 0 & 0 & 2 & 1 & 0 & 2 & 1 \\
    0 & 0 & 1 & 0 & 0 & 0 & 0 & 0 & 0 & 0 & 0 & 0 & 0 & 1 & 0 & 0 & 1 & 2 & 2 & 2 & 0 & 1 & 1 & 1 \\
    0 & 0 & 0 & 1 & 0 & 0 & 0 & 0 & 0 & 0 & 0 & 0 & 2 & 1 & 0 & 0 & 2 & 0 & 0 & 1 & 2 & 1 & 2 & 1 \\
    0 & 0 & 0 & 0 & 1 & 0 & 0 & 0 & 0 & 0 & 0 & 0 & 2 & 2 & 1 & 0 & 2 & 2 & 1 & 1 & 2 & 0 & 0 & 0 \\
    0 & 0 & 0 & 0 & 0 & 1 & 0 & 0 & 0 & 0 & 0 & 0 & 0 & 2 & 1 & 0 & 2 & 0 & 1 & 2 & 0 & 1 & 2 & 2 \\
    0 & 0 & 0 & 0 & 0 & 0 & 1 & 0 & 0 & 0 & 0 & 0 & 2 & 0 & 1 & 2 & 2 & 0 & 1 & 0 & 1 & 1 & 0 & 1 \\
    0 & 0 & 0 & 0 & 0 & 0 & 0 & 1 & 0 & 0 & 0 & 0 & 1 & 2 & 0 & 2 & 1 & 2 & 2 & 1 & 0 & 2 & 0 & 0 \\
    0 & 0 & 0 & 0 & 0 & 0 & 0 & 0 & 1 & 0 & 0 & 0 & 1 & 2 & 2 & 2 & 2 & 0 & 0 & 0 & 0 & 2 & 2 & 1 \\
    0 & 0 & 0 & 0 & 0 & 0 & 0 & 0 & 0 & 1 & 0 & 0 & 1 & 2 & 2 & 1 & 0 & 0 & 2 & 2 & 0 & 1 & 1 & 0 \\
    0 & 0 & 0 & 0 & 0 & 0 & 0 & 0 & 0 & 0 & 1 & 0 & 2 & 1 & 0 & 1 & 2 & 1 & 2 & 1 & 0 & 0 & 1 & 0 \\
    0 & 0 & 0 & 0 & 0 & 0 & 0 & 0 & 0 & 0 & 0 & 1 & 2 & 2 & 2 & 2 & 1 & 2 & 0 & 0 & 1 & 0 & 0 & 2 \\
\end{bmatrix} \,,
\end{equation}
which is the matrix for $\BF_3$, $n=24$, $d=6$ in \cite{database}.

In this case, the elliptic genus $Z_{\mathrm{EG}}$ depends on the choice of $\mathrm{U}(1)$ current, i.e., characteristic vector $\chi\in\CX_{n/3}$. There are 72 characteristic vectors in $\CX_{n/3}$, which are divided into two types.

For 24 characteristic vectors, e.g.,
\begin{equation}
    \chi_1 = \frac{1}{\sqrt{3}} (1,1,1,1,1,-1,1,1,-1,1,1,-1,1,-1,1,-1,-1,-1,-1,-1,-1,1,-1,1)\,,
\end{equation}
the elliptic genus is
\begin{equation}
\begin{aligned}
    Z_{\mathrm{EG}}(\tau,z; \Lambda(C),\chi_1) &= y^4 - 24y + 22 - 24y^{-1} + y^{-4} + \CO(q) \\[0.1cm]
    &= -\frac{1}{864}\phi_{0,1}^4 + \frac{7}{48}\phi_{0,1}^2 \phi_{-2,1}^2 E_4 + \frac{41}{108}\phi_{0,1} \phi_{-2,1}^3 E_6 + \frac{137}{288} \phi_{-2,1}^4 E_4^2 \\[0.2cm]
    &= 23 \, \mathrm{Ch}_{\frac{7}{2};1,0} + 24 (\mathrm{Ch}_{\frac{7}{2};1,1} + \mathrm{Ch}_{\frac{7}{2};1,-1}) + \mathrm{Ch}_{\frac{7}{2};1,4} \\[0.1cm]
    &\quad + \sum_{t=1}^\infty \left( \sum_{k=1}^3 A_{t,k} (\mathrm{Ch}_{\frac{7}{2};1+t,k} + \mathrm{Ch}_{\frac{7}{2};1+t,-k}) + A_{t,4} \, \mathrm{Ch}_{\frac{7}{2};1+t,4} \right)
\end{aligned}
\end{equation}
where $A_{t,1}=\{32866,2969208,\dots\}$, $A_{t,2}=\{14192,1659246,\dots\}$, $A_{t,3}=\{1976,484929,\dots\}$, and $A_{t,4}=\{47,62056,\dots\}$ for $t=\{1,2,\dots\}$.
The difference with the extremal one in the polar region is the term $q^0y$, thus this is a $\beta$-extremal elliptic genus with $\beta>1$.

For 48 characteristic vectors, e.g.,
\begin{equation}
    \chi_2 = \frac{1}{\sqrt{3}} (1,1,1,1,-1,1,1,1,1,1,-1,-1,1,1,-1,1,-1,-1,1,-1,1,1,-1,1)\,,
\end{equation}
the elliptic genus is
\begin{equation}
\begin{aligned}
    Z_{\mathrm{EG}}(\tau,z; \Lambda(C),\chi_2) &= y^4 + y^2 - 12y + 44 - 12y^{-1} + y^{-2} + y^{-4} + \CO(q) \\[0.1cm]
    &= \frac{1}{864}\phi_{0,1}^4 + \frac{5}{36}\phi_{0,1}^2 \phi_{-2,1}^2 E_4 + \frac{83}{216}\phi_{0,1} \phi_{-2,1}^3 E_6 + \frac{137}{288} \phi_{-2,1}^4 E_4^2 \\[0.2cm]
    &= 45 \, \mathrm{Ch}_{\frac{7}{2};1,0} + 12 (\mathrm{Ch}_{\frac{7}{2};1,1} + \mathrm{Ch}_{\frac{7}{2};1,-1}) + (\mathrm{Ch}_{\frac{7}{2};1,2} + \mathrm{Ch}_{\frac{7}{2};1,-2}) + \mathrm{Ch}_{\frac{7}{2};1,4} \\[0.1cm]
    &\quad + \sum_{t=1}^\infty \left( \sum_{k=1}^3 A_{t,k} (\mathrm{Ch}_{\frac{7}{2};1+t,k} + \mathrm{Ch}_{\frac{7}{2};1+t,-k}) + A_{t,4} \, \mathrm{Ch}_{\frac{7}{2};1+t,4} \right)
\end{aligned}
\end{equation}
where $A_{t,1}=\{32922,2969505,\dots\}$, $A_{t,2}=\{14125,1658880,\dots\}$, $A_{t,3}=\{1989,484920,\dots\}$, and $A_{t,4}=\{45,62120,\dots\}$ for $t=\{1,2,\dots\}$.
The difference with the extremal one in the polar region is the terms $q^0y^2$ and $q^0y$, thus this is a $\beta$-extremal elliptic genus with $\beta>4$. This means that the earlier case is nearer to the extremal elliptic genus.

From the generator expansion, it can be seen that the absolute values of the coefficients of $\phi_{0,1}^4$ are equal in both cases. It can be understood as follows.

Let $\phi(\tau,z)$ be a weak Jacobi form of weight $w=0$ and index $m\in\BZ$. Since $\phi_{0,1}(\tau,0)=12$ and $\phi_{-2,1}(\tau,0)=0$, $\phi(\tau,0)$ is a constant determined by the coefficient of $\phi_{0,1}^m$. On the other hand, $Z_{\mathrm{EG}}(\tau,0\,;\Lambda(C),\chi) = Z_{\widetilde{\mathrm{R}}}(\tau\,;\Lambda(C),\chi)$ and the dependence of $Z_{\widetilde{\mathrm{R}}}$ on $\chi$ is only the overall sign. Thus, the absolute value of the coefficient of $\phi_{0,1}^m$ in $Z_{\mathrm{EG}}(\tau,z\,;\Lambda(C),\chi)$ does not depend on the choice of characteristic vector $\chi$.

We can also observe that the coefficients of $\phi_{-2,1}^4 E_4^2$ are equal in both cases. This is not a coincidence.
In what follows, we show that the coefficient of $\phi_{-2,1}^{n/6} E_4^{n/12} $ in the U(1)-graded partition function~\eqref{eq:elliptic_from_lattice} does not depend on the choice of characteristic vectors $\chi\in\CX_{n/3}$.

Let $\phi(\tau,z)$ be a weak Jacobi form of weight $w$ and index $m\in\BZ$. Substituting $a=d=0,\ b=-1,\ c=1$ and $\tau=\i$ in \eqref{eq:weak_Jacobi_md}, we get
\begin{equation}
    \phi(\i, -\i z) = \i^w e^{2\pi m z^2} \phi(\i,z) \,.
\end{equation}
When $z=\frac{1}{2}(1-\i)$, from $-\i z=z-1$ and \eqref{eq:weak_Jacobi_sf}, 
\begin{equation}
    \phi(\i,\tfrac{1}{2}(1-\i)) = \i^w e^{-\i\pi m} \phi(\i,\tfrac{1}{2}(1-\i)) \,.
\end{equation}
Thus, $\phi(\i,\frac{1}{2}(1-\i))$ can be nonzero only if $w+2m\in4\BZ$. For the four generators, it follows that $E_6(\i)=0$ and $\phi_{0,1}(\i,\frac{1}{2}(1-\i))=0$. Then, when a weak Jacobi form $\phi(\tau,z)$ such that $w+2m\in4\BZ$ is decomposed by generators, only the term consisting of $E_4(\tau)$ and $\phi_{-2,1}(\tau,z)$, i.e. $(E_4)^{(w+2m)/4} (\phi_{-2,1})^m$, can have nonzero contribution at $\tau=\i,\ z=\frac{1}{2}(1-\i)$. In other words, the coefficient of such a term is determined solely by the value $\phi(\i,\frac{1}{2}(1-\i))$.
On the other hand, from \eqref{eq:elliptic_from_lattice}, the elliptic genus at that values is
\begin{equation}
\begin{aligned}
    Z_{\mathrm{EG}}(\i,\tfrac{1}{2}(1-\i)\,;\Lambda(C),\chi) &= \frac{1}{\eta(\i)^n}\,\sum_{\lambda\,\in\,\Lambda(C)} e^{\pi\chi\cdot (\lambda+\frac{\chi}{2})} e^{-\pi(\lambda+\frac{\chi}{2})^2} \\
    &= \frac{1}{\eta(\i)^n}\,\sum_{\lambda\,\in\,\Lambda(C)} e^{-\pi(\lambda^2-\frac{\chi^2}{4})} \,.
\end{aligned}
\end{equation}
If we fix $\chi^2=n/3$, this does not depend on the choice of $\chi$.
Combining them, we can conclude that the coefficient of $E_4^{n/12} \phi_{-2,1}^{n/6}$ in $Z_{\mathrm{EG}}(\tau,z\,;\Lambda(C),\chi)$ does not depend on $\chi$.

\subsection{$n=36\,,\,$ $p=3$}

For the central charge 36, the polar terms of the extremal elliptic genus are
\begin{align}
    y^6 + 0 (y^5 + y^4 + y^3 + y^2 + y) + q (y^6 - y^5) \,.
\end{align}
However, as shown in section~\ref{ss:extremal_intro}, the extremal elliptic genus cannot be realized as a weak Jacobi form. Moreover, it is not possible to construct the elliptic genus such that the coefficient of $q y^6$ is $1$ from codes over any $\BF_p$ since if we choose the $\mathrm{U}(1)$ current as $J(z)=\chi\cdot L(z)$, $\chi\in\CX_{n/3}$, then the state $\ket{\chi/2}$ in the R sector corresponds to $q^0y^6$ and its descendants $\alpha_{-1}^i \ket{\chi/2},\,i=1,\dots,36$ to $36qy^6$ in the elliptic genus.

From the discussion at $n=24$, it is natural to expect that a code with a large minimum Hamming weight gives a near-extremal elliptic genus. In the ternary case, the largest minimum Hamming weight that a self-dual code with length $n=36$ can have is known to be $12$. Thus, we consider a code $C\subset\BF_3^{36}$ generated by ($[n,k,d]=[36,18,12]$ code in \cite{codetables})
\begin{equation}
\begin{bmatrix}
    1&0&0&0&0&0&0&0&0&0&0&0&0&0&0&0&0&0&0&2&2&1&2&1&1&1&2&2&1&1&1&2&1&2&2&1 \\
    0&1&0&0&0&0&0&0&0&0&0&0&0&0&0&0&0&0&2&0&2&2&1&2&1&1&1&2&2&1&1&1&2&1&2&1 \\
    0&0&1&0&0&0&0&0&0&0&0&0&0&0&0&0&0&0&2&2&0&2&2&1&2&1&1&1&2&2&1&1&1&2&1&1 \\
    0&0&0&1&0&0&0&0&0&0&0&0&0&0&0&0&0&0&1&2&2&0&2&2&1&2&1&1&1&2&2&1&1&1&2&1 \\
    0&0&0&0&1&0&0&0&0&0&0&0&0&0&0&0&0&0&2&1&2&2&0&2&2&1&2&1&1&1&2&2&1&1&1&1 \\
    0&0&0&0&0&1&0&0&0&0&0&0&0&0&0&0&0&0&1&2&1&2&2&0&2&2&1&2&1&1&1&2&2&1&1&1 \\
    0&0&0&0&0&0&1&0&0&0&0&0&0&0&0&0&0&0&1&1&2&1&2&2&0&2&2&1&2&1&1&1&2&2&1&1 \\
    0&0&0&0&0&0&0&1&0&0&0&0&0&0&0&0&0&0&1&1&1&2&1&2&2&0&2&2&1&2&1&1&1&2&2&1 \\
    0&0&0&0&0&0&0&0&1&0&0&0&0&0&0&0&0&0&2&1&1&1&2&1&2&2&0&2&2&1&2&1&1&1&2&1 \\
    0&0&0&0&0&0&0&0&0&1&0&0&0&0&0&0&0&0&2&2&1&1&1&2&1&2&2&0&2&2&1&2&1&1&1&1 \\
    0&0&0&0&0&0&0&0&0&0&1&0&0&0&0&0&0&0&1&2&2&1&1&1&2&1&2&2&0&2&2&1&2&1&1&1 \\
    0&0&0&0&0&0&0&0&0&0&0&1&0&0&0&0&0&0&1&1&2&2&1&1&1&2&1&2&2&0&2&2&1&2&1&1 \\
    0&0&0&0&0&0&0&0&0&0&0&0&1&0&0&0&0&0&1&1&1&2&2&1&1&1&2&1&2&2&0&2&2&1&2&1 \\
    0&0&0&0&0&0&0&0&0&0&0&0&0&1&0&0&0&0&2&1&1&1&2&2&1&1&1&2&1&2&2&0&2&2&1&1 \\
    0&0&0&0&0&0&0&0&0&0&0&0&0&0&1&0&0&0&1&2&1&1&1&2&2&1&1&1&2&1&2&2&0&2&2&1 \\
    0&0&0&0&0&0&0&0&0&0&0&0&0&0&0&1&0&0&2&1&2&1&1&1&2&2&1&1&1&2&1&2&2&0&2&1 \\
    0&0&0&0&0&0&0&0&0&0&0&0&0&0&0&0&1&0&2&2&1&2&1&1&1&2&2&1&1&1&2&1&2&2&0&1 \\
    0&0&0&0&0&0&0&0&0&0&0&0&0&0&0&0&0&1&1&1&1&1&1&1&1&1&1&1&1&1&1&1&1&1&1&0
\end{bmatrix} \,.
\end{equation}

As in the previous example, the elliptic genus $Z_{\mathrm{EG}}$ depends on the choice of $\mathrm{U}(1)$ current, i.e., characteristic vector $\chi\in\CX_{n/3}$. There are 888 characteristic vectors in $\CX_{n/3}$, which are divided into two types.

For 72 characteristic vectors, the elliptic genus is
\begin{equation}
\begin{aligned}
    &Z_{\mathrm{EG}}(\tau,z; \Lambda(C),\chi) \\[0.1cm]
    &= y^6 - 408y + 70 - 408y^{-1} + y^{-6} + (36y^6 - 36y^5 + \dots ) q + \CO(q^2) \\[0.1cm]
    &= -\frac{31}{124416}\phi_{0,1}^6 + \frac{365}{41472}\phi_{0,1}^4 \phi_{-2,1}^2 E_4 + \frac{269}{7776}\phi_{0,1}^3 \phi_{-2,1}^3 E_6 + \frac{3071}{13824} \phi_{0,1}^2 \phi_{-2,1}^4 E_4^2 \\[0.2cm]
    &\quad + \frac{271}{648} \phi_{0,1} \phi_{-2,1}^5 E_4 E_6 + \frac{3281}{13824} \phi_{-2,1}^6 E_4^3 + \frac{1231}{15552} \phi_{-2,1}^6 E_6^2 \\[0.2cm]
    &= 71 \, \mathrm{Ch}_{\frac{11}{2};\frac{3}{2},0} + 408 (\mathrm{Ch}_{\frac{11}{2};\frac{3}{2},1} + \mathrm{Ch}_{\frac{11}{2};\frac{3}{2},-1}) + \mathrm{Ch}_{\frac{11}{2};\frac{3}{2},6} \\
    &\quad + \sum_{t=1}^\infty \left( \sum_{k=1}^5 A_{t,k} (\mathrm{Ch}_{\frac{11}{2};\frac{3}{2}+t,k} + \mathrm{Ch}_{\frac{11}{2};\frac{3}{2}+t,-k}) + A_{t,6} \, \mathrm{Ch}_{\frac{11}{2};\frac{3}{2}+t,6} \right)
\end{aligned}
\end{equation}
where $A_{t,1}=\{\,580550\,,\,167398422\,,\,\dots\,\}$\,, \,$A_{t,2}=\{\,298452\,,\,104775760\,,\,\dots\,\}$\,,\, $A_{t,3}=\{\,71400\,,\\ 39702684\,,\,\dots\}$, $A_{t,4}=\{7140,8410920,\dots\}$, $A_{t,5}=\{0,843115,\dots\}$, and $A_{t,6}=\{35,29154,\dots\}$ for $t=\{1,2,\dots\}$.

For 816 characteristic vectors, the elliptic genus is
\begin{equation}
\begin{aligned}
    &Z_{\mathrm{EG}}(\tau,z; \Lambda(C),\chi) \\[0.1cm]
    &= y^6 + 93y^2 - 36y + 628 - 36y^{-1} + 93y^{-2} + y^{-6} + (36y^6 - 36y^5 + \dots ) q + \CO(q^2) \\[0.1cm]
    &= \frac{31}{124416}\phi_{0,1}^6 + \frac{241}{41472}\phi_{0,1}^4 \phi_{-2,1}^2 E_4 + \frac{569}{15552}\phi_{0,1}^3 \phi_{-2,1}^3 E_6 + \frac{3133}{13824} \phi_{0,1}^2 \phi_{-2,1}^4 E_4^2 \\[0.2cm]
    &\quad + \frac{2137}{5184} \phi_{0,1} \phi_{-2,1}^5 E_4 E_6 + \frac{3281}{13824} \phi_{-2,1}^6 E_4^3 + \frac{631}{7776} \phi_{-2,1}^6 E_6^2 \\[0.2cm]
    &= 629 \, \mathrm{Ch}_{\frac{11}{2};\frac{3}{2},0} + 36 (\mathrm{Ch}_{\frac{11}{2};\frac{3}{2},1} + \mathrm{Ch}_{\frac{11}{2};\frac{3}{2},-1}) + 93 (\mathrm{Ch}_{\frac{11}{2};\frac{3}{2},2} + \mathrm{Ch}_{\frac{11}{2};\frac{3}{2},-2}) + \mathrm{Ch}_{\frac{11}{2};\frac{3}{2},6} \\[0.2cm]
    &\quad + \sum_{t=1}^\infty \left( \sum_{k=1}^5 A_{t,k} (\mathrm{Ch}_{\frac{11}{2};\frac{3}{2}+t,k} + \mathrm{Ch}_{\frac{11}{2};\frac{3}{2}+t,-k}) + A_{t,6} \, \mathrm{Ch}_{\frac{11}{2};\frac{3}{2}+t,6} \right)
\end{aligned}
\end{equation}
where $A_{t,1}=\{\,581480\,,\,167402049\,,\,\dots\,\}$\,,\, $A_{t,2}=\{\,296499\,,\,104768320\,,\,\dots\,\}$\,,\, $A_{t,3}=\{\,72795\,,\\39708264,\dots\}$, $A_{t,4}=\{6768,8409525,\dots\}$, $A_{t,5}=\{0,842185,\dots\}$ and $A_{t,6}=\{35,29712,\dots\}$, for $t=\{1,2,\dots\}$.

They are both $\beta$-extremal elliptic genera with $\beta>12$ since they include the term $36qy^6$. 
In addition, the coefficients of $\phi_{-2,1}^6 E_4^3$ and the absolute values of the coefficients of $\phi_{0,1}^6$ are equal respectively, as shown in the previous example.

\section{Discussion}
\label{sec:discussion}

In this paper, we have considered fermionic CFTs constructed from classical code over $\BF_p$~\cite{Kawabata:2023nlt} and given a systematic way of computing the U(1)-graded partition function that returns elliptic genera for $\CN=2$ theories.
Taking the U(1) current to be consistent with the $\CN=2$ algebra, we have shown that the U(1)-graded partition function exhibits modular invariance and spectral flow invariance.
Using some examples of classical codes, we demonstrated the unified construction of extremal $\CN=2$ elliptic genera for central charges $n=12,24$. Also, we discussed other examples that give near-extremal elliptic genera.

Except for the ternary case, the existence of $\CN=2$ supersymmetry is not manifest in fermionic code CFTs. However, Theorem~\ref{theorem:weak} dictates that, under some assumptions, the U(1)-graded partition function becomes a weak Jacobi form, which is invariant under the modular transformation and spectral flow. Additionally, we can consistently identify the spectral flow operator and chiral primary fields.
These observations naturally lead us to the following conjecture.
\begin{conjecture}
    If $\Lambda$ is an odd self-dual lattice of rank $n\in12\BZ$ and the minimum of characteristic vectors in $\Lambda$ is $n/3$, then the fermionic CFT whose set of vertex operators is specified by $\Lambda$ admits $\CN=2$ supersymmetry.
    \label{conj:super}
\end{conjecture}
If an odd self-dual lattice can be constructed from a ternary self-dual code up to rotation, then the resulting CFT admits $\CN=2$ supersymmetry~\cite{Gaiotto:2018ypj}.
However, as of now, we are not certain whether any lattice satisfying the assumption in Conjecture~\ref{conj:super} accepts the construction from ternary codes.
It would be interesting to pursue the validity of the conjecture.

Whereas this paper focuses on $\CN=2$ supersymmetry in fermionic code CFTs,
there is a possibility that $\CN=4$ supersymmetry emerges in fermionic code CFTs. In this case, the R-symmetry becomes SU(2) symmetry rather than U(1).
Note that, in the binary case, any fermionic code CFT contains affine SU(2) symmetry as discussed in the bosonic case~\cite{Dolan:1994st}. 
Hence, the binary construction might be helpful to provide $\CN=4$ superconformal theories.

Although this paper deals with the construction of extremal $\CN=2$ CFTs from classical codes, it would be interesting if one could construct extremal $\CN=4$ CFTs, some of which have been realized in actual CFTs at small central charges \cite{Gaberdiel:2008xb,Cheng:2014owa,Harrison:2016hbq}.
For the central charge $n=12$, it is known that extremal $\CN=2$ and $\CN=4$ elliptic genera are the same functions, so the fermionic code CFT constructed in section~\ref{ss:ext_m=2} realizes both the extremal $\CN=2$ and $\CN=4$ theories.
For the central charge $n=24$, the extremal $\CN=2$ elliptic genus is different from the $\CN=4$ one.
To construct $\CN=4$ extremal CFT, we attempted to find binary codes $(p=2)$ that return the Ramond-Ramond partition function consistent with the $\CN=4$ extremal elliptic genus~\cite{Harrison:2016hbq}.
However, we could not find such binary codes from the database~\cite{database}.
Also, in the ternary case $(p=3)$, it is straightforward to deduce from Theorem~1.3 in~\cite{Gaiotto:2018ypj} that any ternary code cannot yield the $\CN=4$ extremal elliptic genus.
If one considers classical codes over $p\geq5$, there would be some chance to obtain the $\CN=4$ extremal CFT.
Additionally, there would be some possibilities that $\CN=4$ extremal CFTs arise if one can generalize the construction of fermionic code CFTs to the finite rings $\BZ_k$ for an integer $k$.
Recently, the construction of Narain CFTs from quantum codes over the rings has been given~\cite{Alam:2023qac}.
Analogously, we may construct fermionic CFTs from classical codes over rings. 
We can expect that such an extension widely broadens the scope of fermionic code CFTs and yields theories inaccessible in the present paper.

In section~\ref{sec:extremal}, we constructed extremal $\CN=2$ elliptic genera for central charges $n=12$ and $n=24$. To construct extremal CFTs with larger central charges, we have to care about the descendants of the bosonic oscillators. 
Those descendants always appear in CFTs based on lattices and are obstacles to constructing extremal CFTs.
One possible candidate to resolve the problem is to take the orbifold by an appropriate symmetry as in the construction of the Monster CFT \cite{frenkel1984natural}, which is an extremal bosonic CFT.
As discussed in the ternary case \cite{Gaiotto:2018ypj}, it would be interesting to study the fermionic orbifold of code CFTs.

In section~\ref{sec:extremal} and \ref{sec:near-extremal}, we computed the elliptic genera and decompose them into $\CN=2$ superconformal characters.
The Mathieu moonshine was discovered by noting the relationship between the coefficients of the characters in the elliptic genus and the dimensions of the irreducible representations~\cite{Eguchi:2010ej} (see, e.g.,~\cite{Anagiannis:2018jqf} and references therein).
Our construction of elliptic genera uses classical error-correcting codes, and their automorphisms (symmetries) have been studied extensively~\cite{conway2013sphere}.
There may be an intriguing relationship between automorphisms of classical codes and coefficients of characters in elliptic genera.
It would help acquire a better understanding of the moonshine phenomena.

In the present paper, we computed examples of the U(1)-graded partition function with the Fourier expansion~\eqref{eq:exp_weak}. In other words, the unitarity bound $(h\geq n/24)$ in the R sector holds.
However, one can consider the U(1)-graded partition function whose Fourier expansion has singular terms.
For example, if we take the last example in Appendix B of \cite{Kawabata:2023nlt}, the (supersymmetric) unitarity bound $(h\geq n/24)$ in the R sector does not hold. However, the Construction A lattice contains a characteristic vector with norm $n/3=12$. Therefore, we can construct the U(1)-graded partition function that shows modular invariance and spectral flow invariance while it does not admit the Fourier expansion~\eqref{eq:exp_weak}.
It would be interesting to figure out its property with the help of number theory~\cite{Eichler:1985,Dabholkar:2012nd,DHoker:2022dxx}.

\acknowledgments
We are grateful to Hee-Cheol Kim, Minsung Kim, and Yutaka Matsuo for their valuable discussions.
The work of K.\,K. and S.\,Y. was supported by FoPM, WINGS Program, the University of Tokyo.
The work of K.\,K. was supported by JSPS KAKENHI Grant-in-Aid for JSPS fellows Grant No. 23KJ0436.

\appendix

\section{$\CN=2$ superconformal characters}
\label{app:N2_character}

We list the graded characters of the representations of the $\CN=2$ superconformal algebra. The following notation is consistent with \cite{Cheng:2014owa} up to overall signs.
For more details, see e.g.~\cite{Dobrev:1986hq,Kiritsis:1986rv,Matsuo:1986cj,Eguchi:1988af}.

The elliptic genus of an $\CN=2$ SCFT is invariant under integer spectral flows. Therefore, it should be decomposed into graded characters $\mathrm{Ch^{(R)}}(\tau,z)$ in the R sector defined by
\begin{equation}
    \mathrm{Ch^{(R)}}(\tau,z) = \sum_{\theta\in\frac{1}{2}+\BZ} \mathrm{SF}_{\theta} \, \mathrm{ch^{(NS)}}(\tau,z+\tfrac{1}{2})
\end{equation}
where $\mathrm{SF}_{\theta}$ is defined in \eqref{eq:R_vacuum} and $\mathrm{ch^{(NS)}}(\tau,z)$ is an irreducible character in the NS sector.

The characters are classified into three types: massive, massless, and graviton. Let the central charge $n\in 6\BZ,\, l=\frac{n}{6}-\frac{1}{2},\, q=e^{2\pi\i \tau},\, y=e^{2\pi\i z}$, with $h$ and $Q$ representing the conformal weight and the $\mathrm{U}(1)$ charge, i.e., the eigenvalues of $L_0$ and $J_0$.

\paragraph{Massive}
For the NS sector, the irreducible character of the massive representation with the highest weight $h,\,Q$ is
\begin{equation}
    \mathrm{ch}_{l;h,Q}^{(\mathrm{NS})}(\tau,z) = q^{h-\frac{n}{24}} y^Q \prod_{k=1}^\infty \frac{(1+yq^{k-\frac{1}{2}}) (1+y^{-1}q^{k-\frac{1}{2}})}{(1-q^k)^2} \,.
\end{equation}
The contribution $(1-q^k)^{-1}=1+q^k+q^{2k}+\dots$ comes from $L_{-k}$ and $J_{-k}$, $1+yq^r$ from $G_{-r}^+$, and $1+y^{-1}q^r$ from $G_{-r}^-$.

For the R sector, the graded character of the massive representation with the highest weight $h',\,Q'$ is given by
\begin{equation}
    \mathrm{Ch}_{l;h',Q'}^{(\mathrm{R})}(\tau,z) = \sum_{\theta\in\frac{1}{2}+\BZ} \mathrm{SF}_{\theta} \, \mathrm{ch}_{l;h,Q}^{(\mathrm{NS})}(\tau,z+\tfrac{1}{2})
\end{equation}
where
\begin{equation}
    h' = h + \frac{Q}{2} + \frac{n}{24} \,,\quad
    Q' = \begin{cases}
        Q + \frac{n}{6} & (Q'>0) \\
        Q + \frac{n}{6} - 1 & (Q'<0)
    \end{cases} \,.
\end{equation}
From \eqref{eq:sf_LJ}, we obtain a more explicit form as
\begin{equation}
\begin{aligned}
    &\mathrm{Ch}_{l;h',Q'}^{(\mathrm{R})}(\tau,z) \\
    &= \sum_{\theta\in\frac{1}{2}+\BZ} q^{h-\frac{n}{24}+Q\theta+\frac{n}{6}\theta^2} (-y)^{Q+\frac{n}{3}\theta} \prod_{k=1}^\infty \frac{(1-yq^{k-\frac{1}{2}+\theta}) (1-y^{-1}q^{k-\frac{1}{2}-\theta})}{(1-q^k)^2} \\
    &= \sum_{m\in\BZ} q^{h-\frac{n}{24}+Q(m+\frac{1}{2})+\frac{n}{6}(m+\frac{1}{2})^2-\frac{1}{2}m(m+1)} (-y)^{Q+\frac{n}{3}(m+\frac{1}{2})-m} \prod_{k=1}^\infty \frac{(1-yq^{k}) (1-y^{-1}q^{k-1})}{(1-q^k)^2} \\
    &= (-1)^{Q'+\delta} q^{h'-\frac{n}{24}} y^{Q'+\delta} \prod_{k=1}^\infty \frac{(1-yq^{k}) (1-y^{-1}q^{k-1})}{(1-q^k)^2} \sum_{m\in\BZ} (-1)^m q^{lm^2 + (Q'+\delta-\frac{1}{2})m} y^{2lm}
\end{aligned}
\end{equation}
where $\delta$ is $1$ if $Q'<0$ and $0$ if $Q'>0$.

\paragraph{Massless}
For the NS sector, the irreducible character of the massless representation with the highest weight $h=-\frac{Q}{2},\,Q$ is
\begin{equation}
    \mathrm{ch}_{l;h=-\frac{Q}{2},Q}^{(\mathrm{NS})}(\tau,z) = q^{h-\frac{n}{24}} y^Q \frac{1}{1+y^{-1}q^{\frac{1}{2}}} \prod_{k=1}^\infty \frac{(1+yq^{k-\frac{1}{2}}) (1+y^{-1}q^{k-\frac{1}{2}})}{(1-q^k)^2}
\end{equation}
where $1/(1+y^{-1}q^{\frac{1}{2}})$ comes from $G_{-1/2}^-\ket{h=-\frac{Q}{2},Q}=0$.

For the R sector, the graded character of the massless representation with the highest weight $h'=\frac{n}{24},\,Q'$ is given by
\begin{equation}
    \mathrm{Ch}_{l;h'=\frac{n}{24},Q'}^{(\mathrm{R})}(\tau,z) = \sum_{\theta\in\frac{1}{2}+\BZ} \mathrm{SF}_{\theta} \, \mathrm{ch}_{l;h=-\frac{Q}{2},Q}^{(\mathrm{NS})}(\tau,z+\tfrac{1}{2})
\end{equation}
where $Q'=Q+\frac{n}{6}$.
From a similar discussion as the massive representation,
\begin{equation}
\begin{aligned}
    \mathrm{Ch}_{l;h'=\frac{n}{24},Q'}^{(\mathrm{R})}(\tau,z)
    &= (-1)^{Q'} y^{Q'} \prod_{k=1}^\infty \frac{(1-yq^{k}) (1-y^{-1}q^{k-1})}{(1-q^k)^2} \\
    &\qquad \times \sum_{m\in\BZ} \frac{1}{1-y^{-1}q^{-m}} (-1)^m q^{lm^2 + (Q'-\frac{1}{2})m} y^{2lm} \,.
\end{aligned}
\end{equation}

\paragraph{Graviton}
According to the equation (7.5) in \cite{Cheng:2014owa}, the graded character of the graviton representation with the highest weight $h'=\frac{n}{24},\,Q'=\frac{n}{6}$ in the R sector is given by
\begin{equation}
\begin{aligned}
    \mathrm{Ch}_{l;h'=\frac{n}{24},Q'=\frac{n}{6}}^{(\mathrm{R})}(\tau,z)
    &= q^{-s} \left( \mathrm{Ch}_{l;\frac{n}{24}+s,\frac{n}{6}}^{(\mathrm{R})} + \sum_{k=1}^{\frac{n}{6}-1} (-1)^k \left( \mathrm{Ch}_{l;\frac{n}{24}+s,\frac{n}{6}-k}^{(\mathrm{R})} + \mathrm{Ch}_{l;\frac{n}{24}+s,k-\frac{n}{6}}^{(\mathrm{R})} \right) \right) \\
    &\quad + (-1)^{\frac{n}{6}} \mathrm{Ch}_{l;\frac{n}{24},0}^{(\mathrm{R})}
\end{aligned}
\end{equation}
where $s$ is an arbitrary positive integer.

Using these graded characters, we can decompose the elliptic genus of an $\CN=2$ SCFT as
\begin{equation}
    Z_{\mathrm{EG}}(\tau,z) = A_{0,0}\,\mathrm{Ch}_{l;\frac{n}{24},0}^{(\mathrm{R})} + \sum_{t=0}^{\infty} \left( \sum_{k=1}^{\frac{n}{6}-1} A_{t,k} \left( \mathrm{Ch}_{l;\frac{n}{24}+t,k}^{(\mathrm{R})} + \mathrm{Ch}_{l;\frac{n}{24}+t,-k}^{(\mathrm{R})} \right) + A_{t,\frac{n}{6}}\,\mathrm{Ch}_{l;\frac{n}{24}+t,\frac{n}{6}}^{(\mathrm{R})} \right)
\end{equation}
where $A_{t,k}\in\BZ$.

\bibliographystyle{JHEP}
\bibliography{elliptic}

\providecommand{\href}[2]{#2}\begingroup\raggedright\begin{thebibliography}{10}

\bibitem{frenkel1984natural}
I.~B. Frenkel, J.~Lepowsky, and A.~Meurman, {\it A natural representation of
  the fischer-griess monster with the modular function j as character},  {\em
  Proceedings of the National Academy of Sciences} {\bf 81} (1984), no.~10
  3256--3260.

\bibitem{frenkel1989vertex}
I.~Frenkel, J.~Lepowsky, and A.~Meurman, {\em Vertex operator algebras and the
  Monster}.
\newblock Academic press, 1989.

\bibitem{Dolan:1994st}
L.~Dolan, P.~Goddard, and P.~Montague, {\it {Conformal field theories,
  representations and lattice constructions}},  {\em Commun. Math. Phys.} {\bf
  179} (1996) 61--120, [\href{http://arxiv.org/abs/hep-th/9410029}{{\tt
  hep-th/9410029}}].

\bibitem{conway1979monstrous}
J.~H. Conway and S.~P. Norton, {\it Monstrous moonshine},  {\em Bulletin of the
  London Mathematical Society} {\bf 11} (1979), no.~3 308--339.

\bibitem{Dixon:1988qd}
L.~J. Dixon, P.~H. Ginsparg, and J.~A. Harvey, {\it {Beauty and the Beast:
  Superconformal Symmetry in a Monster Module}},  {\em Commun. Math. Phys.}
  {\bf 119} (1988) 221--241.

\bibitem{Benjamin:2015ria}
N.~Benjamin, E.~Dyer, A.~L. Fitzpatrick, and S.~Kachru, {\it {An extremal
  ${\mathcal{N}}=2$ superconformal field theory}},  {\em J. Phys. A} {\bf 48}
  (2015), no.~49 495401, [\href{http://arxiv.org/abs/1507.00004}{{\tt
  arXiv:1507.00004}}].

\bibitem{Harrison:2016hbq}
S.~M. Harrison, {\it {Extremal chiral $\mathcal N=4$ SCFT with $c=24$}},  {\em
  JHEP} {\bf 11} (2016) 006, [\href{http://arxiv.org/abs/1602.06930}{{\tt
  arXiv:1602.06930}}].

\bibitem{Moore:2023zmv}
G.~W. Moore and R.~K. Singh, {\it {Beauty And The Beast Part 2: Apprehending
  The Missing Supercurrent}},  \href{http://arxiv.org/abs/2309.02382}{{\tt
  arXiv:2309.02382}}.

\bibitem{Gaiotto:2018ypj}
D.~Gaiotto and T.~Johnson-Freyd, {\it {Holomorphic SCFTs with small index}},
  {\em Can. J. Math.} {\bf 74} (2022), no.~2 573--601,
  [\href{http://arxiv.org/abs/1811.00589}{{\tt arXiv:1811.00589}}].

\bibitem{Kawabata:2023nlt}
K.~Kawabata and S.~Yahagi, {\it {Fermionic CFTs from classical codes over
  finite fields}},  {\em JHEP} {\bf 05} (2023) 096,
  [\href{http://arxiv.org/abs/2303.11613}{{\tt arXiv:2303.11613}}].

\bibitem{Dymarsky:2020qom}
A.~Dymarsky and A.~Shapere, {\it {Quantum stabilizer codes, lattices, and
  CFTs}},  {\em JHEP} {\bf 21} (2020) 160,
  [\href{http://arxiv.org/abs/2009.01244}{{\tt arXiv:2009.01244}}].

\bibitem{Kawabata:2022jxt}
K.~Kawabata, T.~Nishioka, and T.~Okuda, {\it {Narain CFTs from qudit stabilizer
  codes}},  {\em SciPost Phys. Core} {\bf 6} (2023) 035,
  [\href{http://arxiv.org/abs/2212.07089}{{\tt arXiv:2212.07089}}].

\bibitem{Alam:2023qac}
Y.~F. Alam, K.~Kawabata, T.~Nishioka, T.~Okuda, and S.~Yahagi, {\it {Narain
  CFTs from nonbinary stabilizer codes}},
  \href{http://arxiv.org/abs/2307.10581}{{\tt arXiv:2307.10581}}.

\bibitem{Dymarsky:2020bps}
A.~Dymarsky and A.~Shapere, {\it {Solutions of modular bootstrap constraints
  from quantum codes}},  {\em Phys. Rev. Lett.} {\bf 126} (2021), no.~16
  161602, [\href{http://arxiv.org/abs/2009.01236}{{\tt arXiv:2009.01236}}].

\bibitem{Henriksson:2022dnu}
J.~Henriksson, A.~Kakkar, and B.~McPeak, {\it {Narain CFTs and quantum codes at
  higher genus}},  {\em JHEP} {\bf 04} (2023) 011,
  [\href{http://arxiv.org/abs/2205.00025}{{\tt arXiv:2205.00025}}].

\bibitem{Dymarsky:2022kwb}
A.~Dymarsky and R.~R. Kalloor, {\it {Fake Z}},  {\em JHEP} {\bf 06} (2023) 043,
  [\href{http://arxiv.org/abs/2211.15699}{{\tt arXiv:2211.15699}}].

\bibitem{Dymarsky:2020pzc}
A.~Dymarsky and A.~Shapere, {\it {Comments on the holographic description of
  Narain theories}},  {\em JHEP} {\bf 10} (2021) 197,
  [\href{http://arxiv.org/abs/2012.15830}{{\tt arXiv:2012.15830}}].

\bibitem{Henriksson:2021qkt}
J.~Henriksson, A.~Kakkar, and B.~McPeak, {\it {Classical codes and chiral CFTs
  at higher genus}},  {\em JHEP} {\bf 05} (2022) 159,
  [\href{http://arxiv.org/abs/2112.05168}{{\tt arXiv:2112.05168}}].

\bibitem{Angelinos:2022umf}
N.~Angelinos, D.~Chakraborty, and A.~Dymarsky, {\it {Optimal Narain CFTs from
  Codes}},  \href{http://arxiv.org/abs/2206.14825}{{\tt arXiv:2206.14825}}.

\bibitem{Dymarsky:2021xfc}
A.~Dymarsky and A.~Sharon, {\it {Non-rational Narain CFTs from codes over
  F$_{4}$}},  {\em JHEP} {\bf 11} (2021) 016,
  [\href{http://arxiv.org/abs/2107.02816}{{\tt arXiv:2107.02816}}].

\bibitem{Buican:2021uyp}
M.~Buican, A.~Dymarsky, and R.~Radhakrishnan, {\it {Quantum codes, CFTs, and
  defects}},  {\em JHEP} {\bf 03} (2023) 017,
  [\href{http://arxiv.org/abs/2112.12162}{{\tt arXiv:2112.12162}}].

\bibitem{Furuta:2022ykh}
Y.~Furuta, {\it {Relation between spectra of Narain CFTs and properties of
  associated boolean functions}},  {\em JHEP} {\bf 09} (2022) 146,
  [\href{http://arxiv.org/abs/2203.11643}{{\tt arXiv:2203.11643}}].

\bibitem{Yahagi:2022idq}
S.~Yahagi, {\it {Narain CFTs and error-correcting codes on finite fields}},
  {\em JHEP} {\bf 08} (2022) 058, [\href{http://arxiv.org/abs/2203.10848}{{\tt
  arXiv:2203.10848}}].

\bibitem{Henriksson:2022dml}
J.~Henriksson and B.~McPeak, {\it {Averaging over codes and an $SU(2)$ modular
  bootstrap}},  \href{http://arxiv.org/abs/2208.14457}{{\tt arXiv:2208.14457}}.

\bibitem{Furuta:2023xwl}
Y.~Furuta, {\it {On the Rationality and the Code Structure of a Narain CFT, and
  the Simple Current Orbifold}},  \href{http://arxiv.org/abs/2307.04190}{{\tt
  arXiv:2307.04190}}.

\bibitem{Kawabata:2023usr}
K.~Kawabata, T.~Nishioka, and T.~Okuda, {\it {Supersymmetric conformal field
  theories from quantum stabilizer codes}},
  \href{http://arxiv.org/abs/2307.14602}{{\tt arXiv:2307.14602}}.

\bibitem{Kawabata:2023iss}
K.~Kawabata, T.~Nishioka, and T.~Okuda, {\it {Narain CFTs from quantum codes
  and their $\mathbb{Z}_2$ gauging}},
  \href{http://arxiv.org/abs/2308.01579}{{\tt arXiv:2308.01579}}.

\bibitem{leech1971sphere}
J.~Leech and N.~Sloane, {\it Sphere packings and error-correcting codes},  {\em
  Canadian Journal of Mathematics} {\bf 23} (1971), no.~4 718--745.

\bibitem{conway2013sphere}
J.~H. Conway and N.~J.~A. Sloane, {\em Sphere packings, lattices and groups},
  vol.~290.
\newblock Springer Science \& Business Media, 2013.

\bibitem{Witten:1982df}
E.~Witten, {\it {Constraints on Supersymmetry Breaking}},  {\em Nucl. Phys. B}
  {\bf 202} (1982) 253.

\bibitem{Witten:1982im}
E.~Witten, {\it {Supersymmetry and Morse theory}},  {\em J. Diff. Geom.} {\bf
  17} (1982), no.~4 661--692.

\bibitem{Witten:1986bf}
E.~Witten, {\it {Elliptic Genera and Quantum Field Theory}},  {\em Commun.
  Math. Phys.} {\bf 109} (1987) 525.

\bibitem{serre2012course}
J.-P. Serre, {\em A course in arithmetic}, vol.~7.
\newblock Springer Science \& Business Media, 2012.

\bibitem{milnor1973symmetric}
J.~W. Milnor and D.~Husemoller, {\em Symmetric bilinear forms}, vol.~5.
\newblock Springer, 1973.

\bibitem{Schwimmer:1986mf}
A.~Schwimmer and N.~Seiberg, {\it {Comments on the N=2, N=3, N=4 Superconformal
  Algebras in Two-Dimensions}},  {\em Phys. Lett. B} {\bf 184} (1987) 191--196.

\bibitem{Lerche:1989uy}
W.~Lerche, C.~Vafa, and N.~P. Warner, {\it {Chiral Rings in N=2 Superconformal
  Theories}},  {\em Nucl. Phys. B} {\bf 324} (1989) 427--474.

\bibitem{Witten:2007kt}
E.~Witten, {\it {Three-Dimensional Gravity Revisited}},
  \href{http://arxiv.org/abs/0706.3359}{{\tt arXiv:0706.3359}}.

\bibitem{hohn2008conformal}
G.~H{\"o}hn, {\it Conformal designs based on vertex operator algebras},  {\em
  Advances in Mathematics} {\bf 217} (2008), no.~5 2301--2335.

\bibitem{Jankiewicz:2006mv}
M.~Jankiewicz and T.~W. Kephart, {\it {Modular invariants and Fischer-Griess
  monster}},  in {\em {26th International Colloquium on Group Theoretical
  Methods in Physics}}, 8, 2006.
\newblock \href{http://arxiv.org/abs/math-ph/0608001}{{\tt math-ph/0608001}}.

\bibitem{Gaberdiel:2008xb}
M.~R. Gaberdiel, S.~Gukov, C.~A. Keller, G.~W. Moore, and H.~Ooguri, {\it
  {Extremal N=(2,2) 2D Conformal Field Theories and Constraints of
  Modularity}},  {\em Commun. Num. Theor. Phys.} {\bf 2} (2008) 743--801,
  [\href{http://arxiv.org/abs/0805.4216}{{\tt arXiv:0805.4216}}].

\bibitem{Cheng:2014owa}
M.~C.~N. Cheng, X.~Dong, J.~F.~R. Duncan, S.~Harrison, S.~Kachru, and T.~Wrase,
  {\it {Mock Modular Mathieu Moonshine Modules}},  {\em Res. Math. Sci.} {\bf
  2} (2015) 13, [\href{http://arxiv.org/abs/1406.5502}{{\tt arXiv:1406.5502}}].

\bibitem{Lerche:1987ca}
W.~Lerche and N.~P. Warner, {\it {Index Theorems in $N=2$ Superconformal
  Theories}},  {\em Phys. Lett. B} {\bf 205} (1988) 471--478.

\bibitem{Kawai:1993jk}
T.~Kawai, Y.~Yamada, and S.-K. Yang, {\it {Elliptic genera and N=2
  superconformal field theory}},  {\em Nucl. Phys. B} {\bf 414} (1994)
  191--212, [\href{http://arxiv.org/abs/hep-th/9306096}{{\tt hep-th/9306096}}].

\bibitem{Dijkgraaf:2000fq}
R.~Dijkgraaf, J.~M. Maldacena, G.~W. Moore, and E.~P. Verlinde, {\it {A Black
  hole Farey tail}},  \href{http://arxiv.org/abs/hep-th/0005003}{{\tt
  hep-th/0005003}}.

\bibitem{Moore:2004fg}
G.~W. Moore, {\it {Strings and Arithmetic}},  in {\em {Les Houches School of
  Physics: Frontiers in Number Theory, Physics and Geometry}}, pp.~303--359,
  2007.
\newblock \href{http://arxiv.org/abs/hep-th/0401049}{{\tt hep-th/0401049}}.

\bibitem{Manschot:2007ha}
J.~Manschot and G.~W. Moore, {\it {A Modern Farey Tail}},  {\em Commun. Num.
  Theor. Phys.} {\bf 4} (2010) 103--159,
  [\href{http://arxiv.org/abs/0712.0573}{{\tt arXiv:0712.0573}}].

\bibitem{Manschot:2008zb}
J.~Manschot, {\it {On the space of elliptic genera}},  {\em Commun. Num. Theor.
  Phys.} {\bf 2} (2008) 803--833, [\href{http://arxiv.org/abs/0805.4333}{{\tt
  arXiv:0805.4333}}].

\bibitem{Tachikawa:2023nne}
Y.~Tachikawa, M.~Yamashita, and K.~Yonekura, {\it {Remarks on mod-2 elliptic
  genus}},  \href{http://arxiv.org/abs/2302.07548}{{\tt arXiv:2302.07548}}.

\bibitem{Witten:1993jg}
E.~Witten, {\it {On the Landau-Ginzburg description of N=2 minimal models}},
  {\em Int. J. Mod. Phys. A} {\bf 9} (1994) 4783--4800,
  [\href{http://arxiv.org/abs/hep-th/9304026}{{\tt hep-th/9304026}}].

\bibitem{Hori:2003ic}
K.~Hori, S.~Katz, A.~Klemm, R.~Pandharipande, R.~Thomas, C.~Vafa, R.~Vakil, and
  E.~Zaslow, {\em {Mirror symmetry}}, vol.~1 of {\em Clay mathematics
  monographs}.
\newblock AMS, Providence, USA, 2003.

\bibitem{Eichler:1985}
M.~Eichler and D.~Zagier, {\em The Theory of Jacobi Forms}.
\newblock Birkhäuser Boston, 1985.

\bibitem{niebur1974construction}
D.~Niebur, {\it Construction of automorphic forms and integrals},  {\em
  Transactions of the American Mathematical Society} {\bf 191} (1974) 373--385.

\bibitem{Keller:2020rwi}
C.~A. Keller and J.~M. Quinones, {\it {On the Space of Slow Growing Weak Jacobi
  Forms}},  \href{http://arxiv.org/abs/2011.02611}{{\tt arXiv:2011.02611}}.

\bibitem{macwilliams1977theory}
F.~J. MacWilliams and N.~J.~A. Sloane, {\em The theory of error-correcting
  codes}, vol.~16.
\newblock Elsevier, 1977.

\bibitem{thompson1983error}
T.~M. Thompson, {\em From error-correcting codes through sphere packings to
  simple groups}.
\newblock No.~21. Mathematical Association of America, 1983.

\bibitem{welsh1988codes}
D.~Welsh, {\em Codes and cryptography}.
\newblock Oxford University Press, 1988.

\bibitem{elkies2000lattices1}
N.~D. Elkies, {\it Lattices, linear codes, and invariants, part i},  {\em
  Notices of the AMS} {\bf 47} (2000).

\bibitem{elkies2000lattices}
N.~D. Elkies, {\it Lattices, linear codes, and invariants, part ii},  {\em
  Notices of the AMS} {\bf 47} (2000).

\bibitem{justesen2004course}
J.~Justesen and T.~H{\o}holdt, {\em A course in error-correcting codes},
  vol.~1.
\newblock European Mathematical Society, 2004.

\bibitem{nebe2006self}
G.~Nebe, E.~M. Rains, and N.~J.~A. Sloane, {\em Self-dual codes and invariant
  theory}, vol.~17.
\newblock Springer, 2006.

\bibitem{elkies1999lattices}
N.~D. Elkies, {\it Lattices and codes with long shadows},
  \href{http://arxiv.org/abs/math/9906086}{{\tt math/9906086}}.

\bibitem{Lerche:1988np}
W.~Lerche, A.~N. Schellekens, and N.~P. Warner, {\it {Lattices and Strings}},
  {\em Phys. Rept.} {\bf 177} (1989) 1.

\bibitem{kac1998vertex}
V.~G. Kac, {\em Vertex algebras for beginners}.
\newblock No.~10. American Mathematical Soc., 1998.

\bibitem{BoyleSmith:2023xkd}
P.~Boyle~Smith, Y.-H. Lin, Y.~Tachikawa, and Y.~Zheng, {\it {Classification of
  chiral fermionic CFTs of central charge $\le 16$}},
  \href{http://arxiv.org/abs/2303.16917}{{\tt arXiv:2303.16917}}.

\bibitem{Rayhaun:2023pgc}
B.~C. Rayhaun, {\it {Bosonic Rational Conformal Field Theories in Small Genera,
  Chiral Fermionization, and Symmetry/Subalgebra Duality}},
  \href{http://arxiv.org/abs/2303.16921}{{\tt arXiv:2303.16921}}.

\bibitem{Hohn:2023auw}
G.~H\"ohn and S.~M\"oller, {\it {Classification of Self-Dual Vertex Operator
  Superalgebras of Central Charge at Most 24}},
  \href{http://arxiv.org/abs/2303.17190}{{\tt arXiv:2303.17190}}.

\bibitem{conway1990new_code}
J.~H. Conway and N.~J. Sloane, {\it A new upper bound on the minimal distance
  of self-dual codes},  {\em IEEE Transactions on Information Theory} {\bf 36}
  (1990), no.~6 1319--1333.

\bibitem{Gaiotto:2007xh}
D.~Gaiotto and X.~Yin, {\it {Genus two partition functions of extremal
  conformal field theories}},  {\em JHEP} {\bf 08} (2007) 029,
  [\href{http://arxiv.org/abs/0707.3437}{{\tt arXiv:0707.3437}}].

\bibitem{Gaberdiel:2007ve}
M.~R. Gaberdiel, {\it {Constraints on extremal self-dual CFTs}},  {\em JHEP}
  {\bf 11} (2007) 087, [\href{http://arxiv.org/abs/0707.4073}{{\tt
  arXiv:0707.4073}}].

\bibitem{Yin:2007gv}
X.~Yin, {\it {Partition Functions of Three-Dimensional Pure Gravity}},  {\em
  Commun. Num. Theor. Phys.} {\bf 2} (2008) 285--324,
  [\href{http://arxiv.org/abs/0710.2129}{{\tt arXiv:0710.2129}}].

\bibitem{Yin:2007at}
X.~Yin, {\it {On Non-handlebody Instantons in 3D Gravity}},  {\em JHEP} {\bf
  09} (2008) 120, [\href{http://arxiv.org/abs/0711.2803}{{\tt
  arXiv:0711.2803}}].

\bibitem{Maloney:2007ud}
A.~Maloney and E.~Witten, {\it {Quantum Gravity Partition Functions in Three
  Dimensions}},  {\em JHEP} {\bf 02} (2010) 029,
  [\href{http://arxiv.org/abs/0712.0155}{{\tt arXiv:0712.0155}}].

\bibitem{Gaiotto:2008jt}
D.~Gaiotto, {\it {Monster symmetry and Extremal CFTs}},  {\em JHEP} {\bf 11}
  (2012) 149, [\href{http://arxiv.org/abs/0801.0988}{{\tt arXiv:0801.0988}}].

\bibitem{Li:2008dq}
W.~Li, W.~Song, and A.~Strominger, {\it {Chiral Gravity in Three Dimensions}},
  {\em JHEP} {\bf 04} (2008) 082, [\href{http://arxiv.org/abs/0801.4566}{{\tt
  arXiv:0801.4566}}].

\bibitem{Gaberdiel:2008pr}
M.~R. Gaberdiel and C.~A. Keller, {\it {Modular differential equations and null
  vectors}},  {\em JHEP} {\bf 09} (2008) 079,
  [\href{http://arxiv.org/abs/0804.0489}{{\tt arXiv:0804.0489}}].

\bibitem{Gaberdiel:2010jf}
M.~R. Gaberdiel, C.~A. Keller, and R.~Volpato, {\it {Genus Two Partition
  Functions of Chiral Conformal Field Theories}},  {\em Commun. Num. Theor.
  Phys.} {\bf 4} (2010) 295--364, [\href{http://arxiv.org/abs/1002.3371}{{\tt
  arXiv:1002.3371}}].

\bibitem{Benjamin:2016aww}
N.~Benjamin, E.~Dyer, A.~L. Fitzpatrick, A.~Maloney, and E.~Perlmutter, {\it
  {Small Black Holes and Near-Extremal CFTs}},  {\em JHEP} {\bf 08} (2016) 023,
  [\href{http://arxiv.org/abs/1603.08524}{{\tt arXiv:1603.08524}}].

\bibitem{Bae:2016yna}
J.-B. Bae, K.~Lee, and S.~Lee, {\it {Bootstrapping Pure Quantum Gravity in
  AdS3}},  \href{http://arxiv.org/abs/1610.05814}{{\tt arXiv:1610.05814}}.

\bibitem{Ferrari:2017kbp}
F.~Ferrari and S.~M. Harrison, {\it {Properties of extremal CFTs with small
  central charge}},  \href{http://arxiv.org/abs/1710.10563}{{\tt
  arXiv:1710.10563}}.

\bibitem{database}
M.~Harada and A.~Munemasa, ``Database of self-dual codes.''
\newblock \url{https://www.math.is.tohoku.ac.jp/~munemasa/selfdualcodes.htm}.

\bibitem{lattice_catalogue}
G.~Nebe and N.~Sloane, ``A catalogue of lattices.''
\newblock \url{http://www.math.rwth-aachen.de/~Gabriele.Nebe/LATTICES/}.

\bibitem{o1944construction}
R.~O’Connor and G.~Pall, {\it The construction of integral quadratic forms of
  determinant 1},  {\em Duke Mathematical Journal} {\bf 11} (1944), no.~2
  319--331.

\bibitem{codetables}
M.~Grassl, ``Code tables.''
\newblock \url{http://www.codetables.de/}.

\bibitem{Eguchi:2010ej}
T.~Eguchi, H.~Ooguri, and Y.~Tachikawa, {\it {Notes on the K3 Surface and the
  Mathieu group $M_{24}$}},  {\em Exper. Math.} {\bf 20} (2011) 91--96,
  [\href{http://arxiv.org/abs/1004.0956}{{\tt arXiv:1004.0956}}].

\bibitem{Anagiannis:2018jqf}
V.~Anagiannis and M.~C.~N. Cheng, {\it {TASI Lectures on Moonshine}},  {\em
  PoS} {\bf TASI2017} (2018) 010, [\href{http://arxiv.org/abs/1807.00723}{{\tt
  arXiv:1807.00723}}].

\bibitem{Dabholkar:2012nd}
A.~Dabholkar, S.~Murthy, and D.~Zagier, {\it {Quantum Black Holes, Wall
  Crossing, and Mock Modular Forms}},
  \href{http://arxiv.org/abs/1208.4074}{{\tt arXiv:1208.4074}}.

\bibitem{DHoker:2022dxx}
E.~D'Hoker and J.~Kaidi, {\it {Lectures on modular forms and strings}},
  \href{http://arxiv.org/abs/2208.07242}{{\tt arXiv:2208.07242}}.

\bibitem{Dobrev:1986hq}
V.~K. Dobrev, {\it {Characters of the Unitarizable Highest Weight Modules Over
  the $N=2$ Superconformal Algebras}},  {\em Phys. Lett. B} {\bf 186} (1987)
  43.

\bibitem{Kiritsis:1986rv}
E.~Kiritsis, {\it {Character Formulae and the Structure of the Representations
  of the $N=1$, $N=2$ Superconformal Algebras}},  {\em Int. J. Mod. Phys. A}
  {\bf 3} (1988) 1871.

\bibitem{Matsuo:1986cj}
Y.~Matsuo, {\it {Character Formula of C \ensuremath{<} 1 Unitary Representation
  of $N=2$ Superconformal Algebra}},  {\em Prog. Theor. Phys.} {\bf 77} (1987)
  793.

\bibitem{Eguchi:1988af}
T.~Eguchi and A.~Taormina, {\it {On the Unitary Representations of $N=2$ and
  $N=4$ Superconformal Algebras}},  {\em Phys. Lett. B} {\bf 210} (1988)
  125--132.

\end{thebibliography}\endgroup
\end{document}